\documentclass[final]{ws-rmp}
\pdfoutput=1
\usepackage{amsmath}
\usepackage{amsfonts}
\usepackage{amssymb}
\usepackage{mathtools}
\usepackage{graphicx}
\usepackage{slashed}
\usepackage{setspace}
\usepackage{color}
\def\catchline{}{}{}{}{}
\let\trimmarks

\def\received#1{\par}
\def\revised#1{~\par}
\def\accepted#1{~\par}
\def\Dsl{\,\raise.15ex\hbox{/}\mkern-13.5mu D}

\newcommand{\D}{\mathcal{D}}
\renewcommand{\L}{\mathcal{L}}

\newcommand{\scalarp}[2]{(#1\,,#2)}
\newcommand{\scalar}[2]{\langle#1\,,#2\rangle}

\newcommand{\norm}[1]{\|#1\|}
\newcommand{\normm}[1]{|\negthinspace\|#1\|\negthinspace|}
\renewcommand{\d}{\mathrm{d}}
\newcommand{\pM}{\partial M}
\renewcommand{\H}{\mathcal{H}}
\newcommand{\Df}{\mathfrak{D}}
\newcommand{\bt}{\boldsymbol{\theta}}

\newcommand{\scalarb}[2]{\langle#1,#2\rangle}

\renewcommand{\H}{\mathcal{H}}
\newcommand{\nsl}{\vec{\gamma}\cdot\vec{n}\,}
\newcommand{\eref}[1]{(\ref{#1})}
\def\Dsl{\,\raise.15ex\hbox{/}\mkern-12.5mu D}
\def\dsl{\,\raise.15ex\hbox{/}\mkern-12.5mu \partial}
\newcommand\minus{%
  \setbox0=\hbox{-}%
  \vcenter{%
    \hrule width\wd0 height \the\fontdimen8\textfont3%
  }%
}

\numberwithin{equation}{section}

\newenvironment{enumeratesteps}{%
		\renewcommand{\theenumi}{{\bf \roman{enumi})}}\begin{enumerate}
		}
		{%
		\end{enumerate} \renewcommand{\theenumi}{\alph{enumi})}
		}

\newenvironment{proof2}[1][]{%
\par\addvspace{12pt plus3pt minus3pt}\global\logotrue%
\noindent{\bf Proof#1.\hskip.5em}\ignorespaces}{%
	\par\iflogo\prbox\par
	\addvspace{12pt plus3pt minus3pt}\fi}

\begin{document}

\markboth{M.~ASOREY, A.P.~BALACHANDRAN  AND J.M.~P\'EREZ--PARDO }
{EDGE STATES AT PHASE BOUNDARIES}

%
\catchline{}{}{}{}{}
%

\title{EDGE STATES AT PHASE BOUNDARIES AND THEIR STABILITY}
\author{M.~ASOREY
}

\address{
{\small Departamento de F{\'{\i}}sica Te\'orica, Facultad de Ciencias
\\Universidad de Zaragoza
E-50009  Zaragoza, Spain } \\ 
\email{asorey@unizar.es
} 
}

\author{A.P.~BALACHANDRAN}

\address{Physics Department, Syracuse University\\ Syracuse, New York 13244-1130,U.S.A. \\
	and \\
Departamento de Fisica, Universidad de los Andes, Bogot\'a, Colombia\\
\email{balachandran38@gmail.com }
	}

\author{J.M.~P\'EREZ--PARDO}
\address{
	INFN-Sezione di Napoli\\ Via Cintia Edificio 6, I--80126 Napoli, Italy.\\
\email{juanma@na.infn.it}
	}
	
\maketitle

\begin{history}
\received{(Day Month Year)}
\revised{(Day Month Year)}
\end{history}

\begin{abstract}
{
We analyse the effects of Robin-like boundary conditions on different quantum field theories of spin 0, 1/2 and 1 on manifolds with boundaries. In particular, we show that these conditions often lead to the appearance of edge states. These states play a significant role in physical phenomena like quantum Hall effect and topological insulators. We prove in a rigorous way the existence of spectral lower bounds on the kinetic term of different Hamiltonians, even in the case of abelian gauge fields where it is a non-elliptic differential operator. This guarantees the stability and consistency of massive field theories with masses larger than the lower bound of the kinetic term. Moreover, we find an upper bound for the deepest edge state. In the case of Abelian gauge theories we analyse a generalisation of Robin boundary conditions. For Dirac fermions we analyse the cases of Atiyah-Patodi-Singer and chiral bag boundary conditions. The explicit dependence of the bounds on the boundary conditions and the size of the system is derived under general assumptions.}
\end{abstract}

\date{}

\section{Introduction}

Quantum fields are the fundamental pillars of high energy physics.
In the last few years, it has also been shown that many new condensed matter systems have an effective
description in terms of relativistic quantum fields. This novel perspective not only
was useful for developing new experiments, but also to understand
some open questions in field theory and many body physics. In particular, the lack of analytic tools to study
strongly correlated quantum systems and the difficulties of
some numerical approaches can now be overcome by using condensed matter
analogue systems.  Moreover, they can be used as  quantum simulators to understand
some non-perturbative effects which were elusive by standard techniques.

One of the main differences between condensed matter and fundamental particle physics is that in condensed matter, the materials where the effective
field theory lives are of finite size and have boundaries, in contrast with the unbounded nature of Minkowski space-time. This difference is crucial and brings some
extra features to the former like the appearance of a particular type of bound states ({\it edge states})
which are localised at the boundary of the system. These states give rise to new effects and 
transport phenomena. In particular, they are responsible for the appearance
of edge currents and edge conductivity in  graphene  and topological insulators \cite{Ha10}. 
Physical effects where edge states also play a
relevant role include the quantum Hall and Casimir effects.

A specially interesting case is when edge states appear at the boundary between two
phases, e.g. metal-superconductor or metal-insulator. The contact between a gapless 
phase and gapped one induce interesting dynamical effects on the boundary which
are enhanced by the quantum theory.

In this paper we address the characterisation of edge states in field theories defined
on Riemannian manifolds  $M$ with metric $g$ and regular, codimension one, boundaries $\partial M$. The metric gives volume forms $\d\mu_g$ on $M$ and $\d\mu_{\tilde{g}}$ on $\pM$ and their associated Hilbert spaces $\H_M$ and $\H_{\pM}$\,.  In the case of bosonic field theories (scalar or gauge theories), the operator which accounts for the modes associated with the kinetic energy
comes from the Laplace-Beltrami operator $-\Delta_M$. Consistency of the quantum field theories
requires that this operator is self-adjoint on $\H_M$ and positive, otherwise the
quantum vacuum is unstable and unitarity of the quantum theory is lost. See \cite{gadella,amc,TRGNair} and references therein in this connection.
In the Euclidean space $M=\mathbb{R}^N$\,, these properties are guaranteed. The operator
$-\Delta_M$ is essentially self-adjoint and positive on smooth functions of compact support
in $\mathbb{R}^N$.

However, on manifolds with boundaries,   it does not always have the required properties. 
In particular, for scalar fields in the domain $\D_M$ of the Laplacian,
\begin{align}
	\scalar{\Phi_1}{-\Delta_M\Phi_2}&=\int_M\d\mu_g\,\overline{\Phi}_1(-\Delta_M\Phi_2)\notag\\
		&=\int_M\d\mu_g\, \overline{\nabla\Phi}_1\cdot\nabla\Phi_2-\int_{\pM}\d\mu_{\tilde{g}} \, \overline{\Phi}_1\bigr|_{\pM}\nabla_{\mathbf{n}}\Phi_2\bigr|_{\pM}\;,\label{eq:greenintro}
\end{align}
where, as mentioned above, $\pM$ is the induced manifold at the boundary, $\d\mu_{g}$ and $\d\mu_{\tilde{g}}$ denote the volume forms on $M$ and $\partial M$ respectively and $\mathbf{n}$ is the outward-drawn normal on $\partial M$. Notice that the second term, unlike the first, can be negative for $\Phi_1=\Phi_2$.

In the case that the boundary is an interface boundary between two phases, we just consider it as a double boundary: one face for each
side with opposite normal derivatives, and again asymmetric boundary conditions on the two faces can give rise to negative contributions.

It has been known for a long time \cite{Ba93, Li63} that such a negative spectrum does occur for the Robin boundary conditions
\begin{equation}
	\nabla_{\mathbf{n}}\Phi={\mu}\Phi\;,
	\label{rbc}
\end{equation}
for positive values ($\mu>0$) of the Robin parameter and also for more general boundary conditions \cite{Grubb74}. In particular, Asorey et al. \cite{As05} proved that, when $\mu$ becomes large and the Dirichlet condition is approached, the negative spectrum recedes to  $-\infty$ and at the same time the corresponding eigenstates get progressively more localised at the boundary. The existence of this phenomenon is shown in \cite{amc,Ib13} for even more general boundary conditions. Numerical procedures for analysing the associated eigenvalue problem are also described therein.\\

The physical interest for these boundary conditions relies in the fact that they arise in a natural way at the interface between materials in different phases \cite{ABPP}. The best known examples come from electromagnetism between the normal and superconductor phases. The non-interacting order parameter field $\Phi$, that in general can be a tensor field, is determined by the spectrum and eigenvectors of $-\Delta + m^2$ instead of those of $-\Delta$\,. The shift in the spectrum provided by the term $m^2$ is such as to lift the negative energy levels to non-negative values. At the same time, the bulk levels, which already have positive eigenvalues for the case $m=0$, acquire a gap of order $m^2$. Topological insulators also share these features.

From a mathematical viewpoint the situation above requires the relevant operator $-\Delta$ to be lower bounded. This is to ensure that the shift by a mass term $m^2$ is able to make the operator $-\Delta + m^2$ positive. Lower boundedness of self-adjoint extensions of elliptic operators on compact manifolds was proved by G.~Grubb \cite{Grubb74}. The case of non-compact domains with compact boundaries  was proved very recently by the same author \cite{Grubb12}. Lower boundedness of general self-adjoint extensions of non-elliptic operators like $\d^*\d$ on one-forms, which describes the electromagnetic field, remains an open problem.

This article is focused on the study of the general structure of edge sates for field theories with scalar, vector and spinor fields. Edge states are associated with the negative energy eigenstates of the relevant operators, cf.\ \cite{amc,Ib13,ABPP}.

First we present an alternative proof of the well-known semiboundedness of the Laplace-Beltrami operator in a compact Riemannian manifold with boundary for Robin boundary conditions, cf.\ \cite{Grubb74, Dan09, Dan13}. This is done in Section \ref{sectionRBC}. The proof is based on a recent approach \cite{ibort13} that can be used for more general boundary conditions than those considered here and that works also in the context of non-compact manifolds with compact boundaries. As pointed out there, generalisation to non-compact boundaries is also possible. The advantage of this proof is that it keeps track of the different geometrical structures that naturally appear in the problem and can be used to provide also upper bounds for the ground states. This is done in Section \ref{sectionedge}.  The approach provides estimations of the shape of the eigenfunctions while the upper bounds coincide with the asymptotic results obtained in \cite{kp15}. When they exist, the large negative energy states are located at the boundary earning the name of edge states. 

Furthermore, the proof can be modified to treat the situation of the non-elliptic operator $\d^*\d$ on one-forms mentioned above. The literature has no result on this problem despite its physical relevance for the electromagnetic field. This is done in Section \ref{sectionEM}. In order to prove lower boundedness of this non-elliptic operator, one has to take into account the infinite dimensional kernel of the operator. This is done by means of a convenient use of gauge invariance. The bounds do not depend on the gauge chosen and are thus completely general. The situation for fermionic fields is considered in Section \ref{sectionAPS} and Section \ref{sectionchiral}. We analyse the edge effects arising from different choices of the boundary conditions. In this case the relevant operator is the Dirac operator whose spectrum is unbounded below and above. Edge states are associated in this case with the appearance of eigenstates with eigenvalue $\lambda$ within the mass gap, i.e. $-m < \lambda < m$\,. Again, a convenient modification of the proof in Section \ref{sectionRBC} is used to show the existence of such states for appropriate boundary conditions. In Section \ref{sectionAPS} we consider  Atiyah-Patodi-Singer boundary conditions. Surprisingly, in contrast with what happens for scalar fields and vector fields, there is a threshold size of the manifold below which the edge states disappear. This leads to qualitatively different behaviour in this situation. In Section \ref{sectionchiral} chiral bag boundary conditions are considered showing similar results.


\section{Edge States in Scalar Theories}
\label{sectionRBC}

Let us consider for simplicity a complex massive scalar free field $\phi$ with mass $m$ on a manifold $M$ with a compact regular boundary $\partial M$. 

The quantum dynamics of the field is affected by the presence of the boundary. The main effect is due to
the contribution to the kinetic term of the quantum Hamiltonian governed by the quadratic differential 
operator $-\Delta+m^2$.
Boundary effects appear through the boundary conditions that test fields have to satisfy to preserve unitarity of
the quantum theory \cite{amc}. The explicit requirement is that  $-\Delta+m^2$ has to be a
self-adjoint, positive operator. This consistency condition imposes severe constraints on the type of boundary conditions
that stable physical  systems must satisfy.

We shall consider the family of general Robin boundary conditions  
\begin{equation}\label{bcp}
	\nabla_{\mathbf{n}}\Phi(p)=\mu(p)\Phi(p)\;,\quad p\in\partial M\;,
\end{equation}
where ${\mu}$ is any smooth and hence bounded function which is not necessarily positive. This includes the  superconducting case, where $\mu={m}>0$, as a special case (see \cite{ABPP}). Although these are not the most general local boundary conditions that one can consider, see \cite{ilpp14}, they are general enough to present the phenomena of edge states associated to the presence of boundaries.

Among the eigenstates  $\Phi_n$ of $-\Delta$ satisfying this boundary condition for positive $\mu(p)$\,, there are states with negative eigenvalues, i.e., $ -\Delta\Phi_n=\lambda_n\Phi_n \,$ with $\lambda_n<0$\,.
As $\mu(p)$ approaches infinity,   the boundary condition approaches Dirichlet boundary condition $\Phi|_{\partial M}=0$, the negative eigenvalues approach minus infinity and the corresponding eigenvectors get localised near $\partial M$ (see \cite{As05} for details).

However, the theory may still be consistent and the vacuum stable because of the presence of the mass term.
Even if  $-\Delta$  has a negative spectrum, the whole physical operator $-\Delta+m^2$ may be positive.
This requires the spectrum of $-\Delta$ induced by the boundary condition function $\mu(p)$ to be bounded from below by $-m^2$.
 
Although, mathematically speaking, ${\mu(p)}$ might be any non--trivial function on the boundary $\partial M$, when modelling the behaviour of a superconducting phase, the fact that the penetration length in the superconductor is small across the whole boundary ({\it skin effect})  requires that ${\mu(p)}>0$ for any $p\in\partial M$. Notice that in any case the boundary condition  \eqref{rbc} is invariant under time--reversal as long as ${m}$ and $\mu(p)$ are real. Hence, like  topological insulators, the system is $T$--invariant, with edge excitations and an incompressible bulk \cite{ABPP}.


\subsection{ Lower bound theorem for  Robin boundary conditions}\label{subsec:lowerbound}

From now on and throughout the rest of the article there will appear several different Hilbert spaces associated to differentiable manifolds. We will denote the space of square integrable functions over a Riemannian manifold $\Xi$ as $\L^2(\Xi)$. The associated scalar product and norm will be given respectively by $\scalar{\cdot}{\cdot}_{\Xi}$ and $\norm{\cdot}_\Xi$\;. Incidentally, the subindices may be dropped when referring to the manifold $M$\,. In addition, we will need to use the Sobolev spaces of order 1 and 2 that we shall denote as $\H^1(\Xi)$ and $\H^2(\Xi)$ respectively. These spaces are the natural spaces to define differential operators of first and second order. Their corresponding norms will be denoted as $\norm{\cdot}_{\H^1(\Xi)}$ and $\norm{\cdot}_{\H^2(\Xi)}$ respectively. We refer to \cite{adams03} for further details on these spaces.

We shall show now that there exists a lower bound, $-\mu^2_0$, for the spectrum of
the Laplace-Beltrami operator with Robin boundary condition defined on a Riemannian manifold $M$ with Riemannian metric $g$. The result is given by the following theorem, which
implies that  the quantum field theory is consistent  whenever $m^2\geq\mu_0^2$\,. 

\begin{theorem}\label{thm:scalarthm} Let $M$ be an oriented Riemannian manifold with regular oriented boundary $\partial M$. 
Let $\mathbf n$ be the outgoing normal vector and $\mu$ a smooth function, both defined at the boundary $\pM$\,. 
The Laplace-Beltrami operator $-\Delta$ restricted to smooth functions satisfying  the Robin boundary conditions 
\begin{equation}\label{bcp2}
	\nabla_{\mathbf{n}}\Phi(p)=\mu(p)\Phi(p)\;,\quad p\in\partial M\;,
\end{equation}
is essentially self-adjoint and bounded from below, i.e. 
	\begin{equation}\label{semibound}
		\scalar{\Phi}{-\Delta\Phi}\geq-\mu_0^2\norm{\Phi}^2 ,
	\end{equation}
with $\mu_0$ a   finite positive constant.
\end{theorem}

Although the existence of the lower bound \eqref{semibound} is known \cite{Grubb74}, for completeness we include the proof here. Our  alternative proof is based on cobordism methods which will be very useful to obtain upper bounds in the next section and can be  extended to non-elliptic operators in the case of gauge theories in Section 3.

In the proof we shall make an intensive use of the quadratic forms associated to the operators, cf.\ \cite{Re72, Da95}. While there is a one-to-one correspondence between semibounded self-adjoint extensions and quadratic forms, the quadratic forms have the advantage that while maintaining the semibounds they can be defined on domains larger than those of the operators themselves. More importantly, in some cases even the boundary conditions are suppressed by passing to the quadratic form. This happens, for instance, in the well-known case of Neumann boundary conditions where the domain of the associated quadratic form is the full Sobolev space of order 1, without restrictions, cf.\ \cite[Theorem 7.2.1]{Da95}.

Let us finally remark that for negative Robin parameters the associated self-adjoint extensions of the Laplace-Beltrami operator are trivially positive. Indeed, the right hand side of \eqref{eq:greenintro} is positive for $\Phi_1=\Phi_2$\,. On the contrary, for positive Robin parameters there is no straightforward way to prove semiboundedness. The reason is that the $\L^2(\pM)$-norm of the trace, i.e. of the restriction to the boundary of a given function at the bulk, cannot be bounded by the $\L^2(M)$-norm of the function. This is one of the conclusions of the well known Lions Trace Theorem, cf.\ \cite{lions72}. The proof of the necessary and sufficient conditions for semiboundness of even order elliptic operators appearing in \cite{Grubb74, Grubb12} relies essentially on the fact that the inverse of the operator with Dirichlet\footnote{Strictly speaking it is the so called hard extension of the operator.} boundary conditions is a compact operator. This clearly does not hold for the non-elliptic operator mentioned above and thus this proof cannot be used in this case.\\

Let us consider Riemann normal coordinates in a collar neighbourhood $\Xi$ of the boundary $\pM$\,. Let $r$ be the radial normal coordinate increasing from $r=R_0-\epsilon$ at the inner boundary of the collar
to  $r=R_0$  at the physical boundary  $\partial M$ of the system. In these coordinates,
\begin{equation}\label{split M S2}
	\Xi=[R_0-\epsilon,R_0]\times \partial M,
\end{equation}
while the Riemannian metric $g$ of $M$ restricted to $\Xi$ can be written as
\begin{equation}\label{split metric S2}
	g_{_{\Xi}}=
		\begin{pmatrix}
			1 & 0 \\ 0 & \tilde{g}(r,\bt ) \\ 
		\end{pmatrix}\;,
\end{equation}
in terms of boundary coordinates $\bt:=\{\theta^i\}$ of $\partial M$.

Due to the splitting of $M$ into $\Xi$ and its complement  $M\backslash\Xi$  in $M$
(see Fig.~\ref{fig:split M}), a new boundary appears. We denote it by $\partial\Xi_{-}$\,. Notice that the boundary of the collar is therefore $\partial\Xi=\pM\cup\partial\Xi_-$. This boundary has to be considered whenever one defines differential operators on $\Xi$ and $M\backslash\Xi$\,. We consider different Hilbert spaces on each manifold associated to this splitting and denote them as follows. The spaces of square integrable functions  are going to be denoted as $\L^2(M)$, $\L^2(\Xi)$ and $\L^2(M\backslash\Xi)$. The Sobolev spaces of order $1$ will be: $\H^1(M)$, $\H^1(\Xi)$ and $\H^1(M\backslash\Xi)$\,. These are going to be the natural spaces to define the quadratic forms. Similarly, the Sobolev spaces of order $2$ will be: $\H^2(M)$, $\H^2(\Xi)$ and $\H^2(M\backslash\Xi)$\,. The scalar products and norms associated to these spaces are determined by the restrictions (pull-backs) of the Riemannian metric $g$ and distinguished by subscripts. We will also need to consider three different Laplace-Beltrami operators $-\Delta_M$, $-\Delta_{\Xi}$ and $-\Delta_{M\backslash\Xi}$ each one defined on a different domain $\D_M$, $\D_{\Xi}$ and $\D_{M\backslash\Xi}$\,.

\begin{figure}
\hspace{3cm}\includegraphics[scale=0.2]{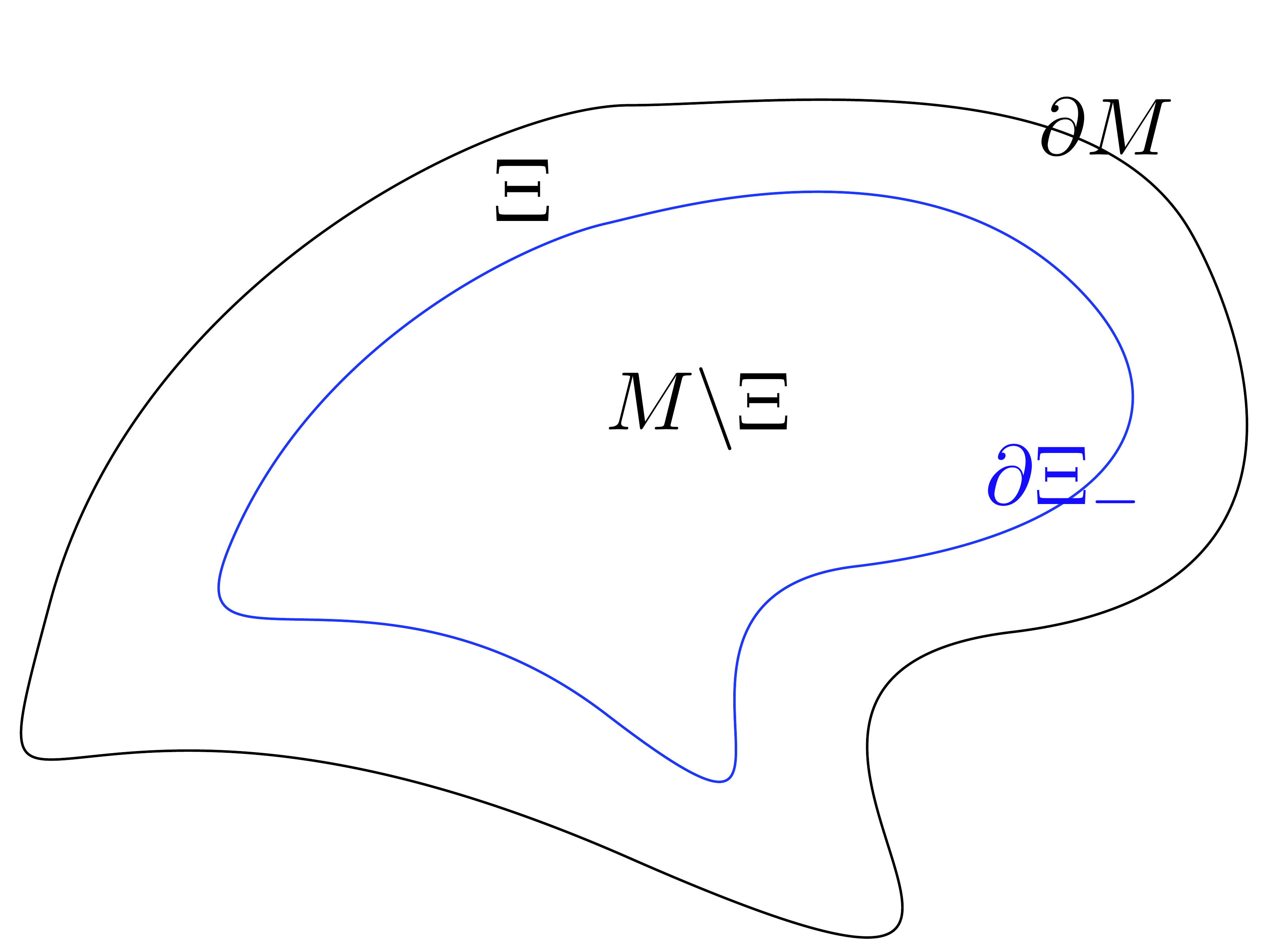}
\caption{Decomposition of the manifold $M$ into the collar neighbourhood $\Xi$ and its complement. $\partial\Xi_{-}$ defines a new boundary.}
\label{fig:split M}
\end{figure}

\begin{proof2}[ of Theorem \ref{thm:scalarthm}]
That the boundary conditions \eqref{bcp2} lead to an essentially self-adjoint operator is a well-known result, see for instance \cite{grubb68,As05,ibort13}. The only restrictions concern the regularity of the function $\mu(p)$\,, but we consider that it is smooth.

The proof of the semiboundedness can be derived in 5 steps.

\begin{enumeratesteps}
\item \label{step:1}{\bf The collar Laplacian.}
Let $-\Delta_\Xi$ be the Laplace-Beltrami operator on the collar $\Xi$ with metric defined by the restriction of $g$\,. Consider as its domain the space of functions that satisfy the Robin boundary condition at $\pM$ and the Neumann boundary condition at the auxiliary boundary $\partial\Xi_{-}$\,, i.e.
\begin{equation}\label{domainXi}
	\D_\Xi=\{\Phi\in\H^2(\Xi)\mid \nabla_{\mathbf{n}}\Phi(p)=\mu(p)\Phi(p)\,,p\in\pM\,; \nabla_{\mathbf{n}}\Phi(p)=0\,,p\in \partial\Xi_{-} \}\;.
\end{equation}
The operator defined this way is essentially self-adjoint.

\item {\bf Bound on the collar Laplacian.}\label{step:2}

The quadratic form associated to the collar Laplacian of the previous step and therefore $-\Delta_{\Xi}$ are bounded from below, i.e.,
\begin{equation}\label{semiboundXi2}
	\scalar{\Phi}{-\Delta_\Xi\Phi}_\Xi\geq-\mu_0^2\norm{\Phi}_{\Xi}^2\,,\quad\Phi\in\overline{\D_\Xi}^{\norm{\cdot}_{\H^1(\Xi)}}\;.
\end{equation}
The constant $\mu_0$ is of the order of $\sup_{p\in\pM}|\mu(p)|$\,. The bar denotes the closure with respect to the corresponding norm. We postpone the proof of this step until the end of this section.

\item {\bf The inner Laplacian.}\label{step:3}
 
Consider the Laplacian $-\Delta_{M\backslash\Xi}$ with Neumann boundary condition on $\partial\Xi_{-}$\,, i.e. defined on the domain
\begin{equation}\label{eq:DMXi}
	\D_{M\backslash\Xi}=\{\Phi\in\H^2({M\backslash\Xi})\mid  \nabla_{\mathbf{n}}\Phi(p)=0\,,p\in \partial\Xi_{-} \}\;.
\end{equation}
This operator is also essentially self-adjoint. It is clearly non-negative since Green's identity gives
$$\scalar{\Phi}{-\Delta_{M\backslash\Xi}\Phi}_{M\backslash\Xi}=\scalar{\nabla\Phi}{\nabla\Phi}_{M\backslash\Xi}\geq 0$$
Moreover, \cite[Theorem 7.2.1]{Da95} ensures that the quadratic form associated to this operator has domain $\H^1(M\backslash\Xi)$ without any constraints.

\item {\bf Equivalence of the bounds of the operator and of the associated quadratic form.}\label{step:4}

We want to find a lower bound for the operator $-\Delta_M$ defined on 
$$\D_M=\{\Phi\in\H^2(M) \mid \nabla_{\mathbf{n}}\Phi(p)=\mu(p) \Phi(p) \,,p\in \pM \}\;,$$
that is, to find a non-negative constant $\mu_0$ such that
$$\scalar{\Phi}{-\Delta_M\Phi}\geq-\mu_0^2\norm{\Phi}^2\,,\quad\Phi\in\D_M\;.$$
The relation between self-adjoint operators and the associated quadratic forms establishes that this is equivalent to showing that
$$\scalar{\Phi}{-\Delta_M\Phi}\geq-\mu_0^2\norm{\Phi}^2\,,\quad\Phi\in\overline{\D_M}^{\normm{\cdot}_{{\Delta_M}}}\;,$$
where $\normm{\cdot}_{\Delta_M}$ stands for the graph-norm of the quadratic form defined by
$$\normm{\Phi}^2_{\Delta_M} := \scalar{\Phi}{\Delta_M\Phi} + m\norm{\Phi}^2\;,$$
where $m$ is a lower bound of the quadratic form.

A sufficient condition for this inequality to hold is 
$$\scalar{\Phi}{-\Delta_M\Phi}\geq-\mu_0^2\norm{\Phi}^2\,,\quad\Phi\in\overline{\D_M}^{\norm{\cdot}_{\H^1(M)}}\;.$$
This is so because the graph-norm of the quadratic form is continuous with respect to the Sobolev norm of order one, i.e.
$$\normm{\Phi}_{\Delta_M}\leq K\norm{\Phi}_{\H^1(M)}\;.$$
To prove the latter inequality it is enough to perform integration by parts once and use the boundary condition to get
\begin{align*}
\scalar{\Phi}{-\Delta_M\Phi}&=\scalar{\nabla\Phi}{\nabla\Phi}-\scalar{\Phi|_{\pM}}{\mu(p)\Phi|_{\pM}}_{\pM}\\
	&\leq \norm{\nabla\Phi}^2+\sup|\mu(p)|\norm{\Phi|_{\pM}}_{\pM}^2\;.
\end{align*}
The first term is clearly bounded by the norm of the Sobolev space of order 1. The second one is bounded by means of the Lions trace inequality, cf. \cite{adams03,lions72}, 
$$\norm{\Phi|_{\pM}}_{\pM}\leq K \norm{\Phi}_{\H^1(M)}\;.$$

\item {\bf Sum of quadratic forms.}\label{step:5}

First notice that because of \cite[Theorem 7.2.1]{Da95}, the Neumann boundary condition at the auxiliary boundary $\partial\Xi_{-}$ disappears when taking the closures with respect to the Sobolev norm of order 1 and therefore we have that
$$\overline{\D_M}^{\norm{\cdot}_{\H^1(M)}}\subset \overline{\D_\Xi}^{\norm{\cdot}_{\H^1(\Xi)}}\oplus\H^1(M\backslash \Xi)\;. $$
The set on the left hand side is strictly contained in the set of the right hand side because functions on the right hand side are allowed to be discontinuous along the auxiliary boundary $\partial\Xi_{-}$ while this is not the case on the left hand side. Now for any $\Phi\in\overline{\D_M}^{\norm{\cdot}_{\H^1(M)}}$\,, we have that
\begin{subequations}\label{eqs:sumqf}
\begin{align}
	\scalar{\Phi}{-\Delta_M\Phi}&=\scalar{\Phi}{-\Delta_{\Xi}\Phi}_{\Xi}+\scalar{\Phi}{-\Delta_{M\backslash\Xi}\Phi}_{M\backslash\Xi}\label{eq:sumqf}\\
		&\geq -\mu_0^2 \norm{\Phi}^2_{\Xi} \geq -\mu_0^2 \norm{\Phi}^2_{M}\;,  \label{eq:M>Xi}
\end{align}
\end{subequations}
where the quadratic forms in \eqref{eq:sumqf} are those defined in Steps~\ref{step:1} and \ref{step:3}. The inequalities of \eqref{eq:M>Xi} follow by Step \ref{step:2} and the fact that 
$$\norm{\Phi}_\Xi^2\leq \norm{\Phi}_\Xi^2 + \norm{\Phi}_{M\backslash\Xi}^2=\norm{\Phi}_M^2\;.$$
\end{enumeratesteps}

\end{proof2}

\begin{proof2}[ of Step \ref{step:2}]

Given a  fixed parameter $0<\delta<1$, the parameter $\epsilon$ defining the width of the collar neighbourhood $\Xi$ can be chosen such that the bounds  
	\begin{equation}\label{compact omega}
		(1-\delta)\sqrt{|\tilde{g}(R_0,\bt)|}\leq\sqrt{|\tilde{g}(r,\bt)|}\leq (1+\delta)\sqrt{|\tilde{g}(R_0,\bt)|}
	\end{equation}
hold.

Now let $I:=[R_0-\epsilon,R_0]$\,. According to the boundary conditions of Eq.~\eqref{domainXi}, we have
	\begin{align}
		\scalar{\Phi}{-\Delta_{\Xi}\Phi}_{\Xi}	&=\int_{\partial M}\int_{I} g_{_{\Xi}}^{-1}(\d\Phi,\d\Phi)\d\mu_{g_{_{\Xi}}}-\scalar{\varphi}{\mu\varphi}_{\partial M}\notag\\
								&\!\!\!\!\!\!\! =\int_{\partial M}\int_{I} \frac{\partial\bar{\Phi}}{\partial r}\frac{\partial\Phi}{\partial r}\d\mu_{g_{_{\Xi}}}+\int_{\partial M}\int_{I}\tilde{g}^{-1}(\d_{\bt}\Phi,\d_{\bt}\Phi)\d\mu_{g_{_{\Xi}}}-\scalar{\varphi}{\mu\varphi}_{\partial M}\notag
								\\
								&\!\!\!\!\!\!\! \geq \int_{\partial M}\int_{I} \frac{\partial\bar{\Phi}}{\partial r}\frac{\partial\Phi}{\partial r} \d\mu_{g_{_{\Xi}}}
								-\scalar{\varphi}{\mu\varphi}_{\partial M}\notag 
								\\
								&\!\!\!\!\!\!\! \geq (1-\delta)\int_{\partial M}\left[\Bigl[\int_{I} \frac{\partial\bar{\Phi}}{\partial r}\frac{\partial\Phi}{\partial r} \d r\Bigr] -\frac{\bar{\mu}}{1-\delta}|\varphi|^2 \right]\d\mu_{\tilde{g}}\;,\label{nada} 
	\end{align}
where $\varphi=\Phi|_{\pM}$ denotes the boundary value of
$\Phi$, while $g_{_{\Xi}}:=g|_{\Xi}$, $\bar{\mu}=\sup \mu(p)$ and $\d_{\bt}\Phi$ denotes the components of $\d\Phi$ that are tangent to the boundary. Notice that there is only a  boundary term  contribution on $\partial M$ because the Neumann boundary condition cancels the boundary contribution of the inner boundary $\partial\Xi_{-}$.
The first inequality holds because the second term of the second line is positive definite. 

The last inequality above ensures that a lower bound of $\scalar{\Phi}{-\Delta_{\Xi}\Phi}_{\Xi}$ can be given by a lower bound of the quadratic form on the right hand side. It is easy to show that this latter quadratic form is the quadratic form associated to the operator 
\begin{equation}\label{eq:radial operator}
	{-\Delta_r}:=-\frac{\d^2}{\d r^2}\otimes\mathbb{I}
\end{equation}
defined on $\L^2(I)\otimes\L^2(\pM)$ by the collar radial projection of the Laplacian. The  operator $-\frac{\d^2}{\d r^2}$ is  densely defined  on $\L^2(I)$  and is essentially self-adjoint with mixed boundary conditions: Neumann boundary condition at $\{R_0-\epsilon\}$ and Robin boundary condition $$\phi'(R_0)=\frac{\bar{\mu}}{1-\delta}\phi(R_0)\;,$$ at $R_0$\,. Moreover the lower bound of $-\Delta_r$ coincides with that of  $-\frac{\d^2}{\d r^2}$. For a more detailed proof we refer to \cite{ibort13}, where more general boundary conditions are considered.

In order to obtain the bound \eqref{semiboundXi2}, it is therefore enough to show that the one-dimensional problem of the Laplacian with Robin boundary condition at one boundary and Neumann boundary condition at the other is bounded from below.

Consider the following one-dimensional spectral boundary value problem for functions $\phi$ 
in the interval $[0,\epsilon]$:
	\begin{equation}\label{laplacian interval2}
		-\phi''=E\phi\;,\quad
		\phi'(0)=0
		\quad
		\phi'(\epsilon)=c\, \phi(\epsilon)\,,\quad c>0\;.
	\end{equation}
We can obtain its spectrum by finding the zeros of the spectral function that is obtained from the plane wave ansatz $\phi(x)=A_+ e^{ i k x}+
A_- e^{ -i k x}$.
The spectral function for any type of boundary condition in $[0,\epsilon]$ has been analysed in \cite{gadella, amc}. In the present case,
it reduces to
	\begin{equation}\label{eq:spectraleq}
		k \tan(\epsilon k)=-c\;.
	\end{equation}
Therefore there is an infinite number of positive eigenvalues $E_n=k^2_n$. 
However, there is a special bound state with negative energy $E_0=-\kappa^2$ which is obtained from
the spectral equation \eqref{eq:spectraleq} assuming that $k=i \kappa$ is imaginary  and $\kappa$ is given by the solution of the equation
	\begin{equation}\label{nrbc}
		\kappa \tanh(\epsilon \kappa)=c\;.
	\end{equation}

This equation has only one solution for positive $c$\,. Thus we have shown that \eqref{eq:radial operator} is  semibounded. Therefore the quadratic form on the right hand side of \eqref{nada} is semibounded  with the same constant. This shows \eqref{semiboundXi2} and therefore Theorem~\ref{thm:scalarthm}.
\vspace{2pt}
\end{proof2}

 \begin{figure}[h]
 \centerline{  \includegraphics[height=6.5cm]{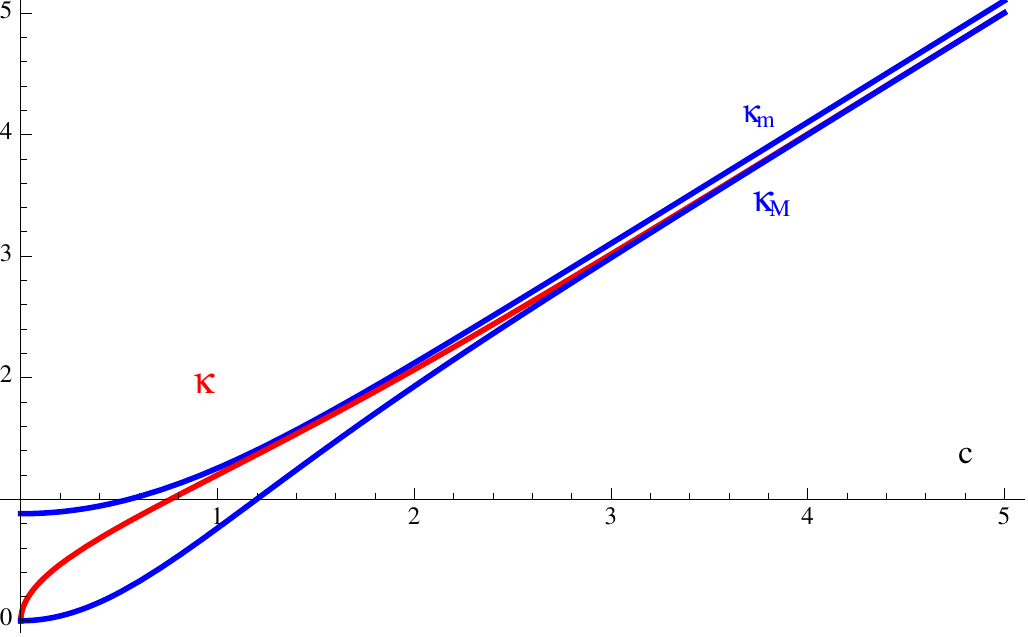}}
 \caption{{Bounds on the negative energy edge mode for a particle on [0,1] with Neumann boundary condition at 0 and
 Robin boundary condition $\phi'(1)=c\phi(1)$\,, $c>0$\,. The red curve is the exact solution and the upper and lower bounds correspond to
 $\kappa_m(c)=(\frac1{10} + \frac{c}{\tanh c})\tanh(\frac1{10} + \frac{c}{\tanh c})$ and $ \kappa_M(c)=c \tanh c$\,.}}\label{fig:2}
  \end{figure}
  
The relations above show the dependence on the various parameters of the problem.

  The solution $\kappa( c)$ of \eref{nrbc} can be found numerically \cite{Ib13}, but it can also be bounded from above, $\kappa(c)<\kappa_m(c)$\,,
by the following function:
  	\begin{equation}\label{bb}
		\kappa_m(c)=\textstyle{\frac1{\epsilon}(\frac1{10} + \frac{ \epsilon c}{\tanh  \epsilon c}})\tanh(\frac1{10} + \frac{\epsilon c}{\tanh \epsilon c})\;.
			\end{equation}
(See Fig.\ \ref{fig:2} where we set $\epsilon=1$\,).

This bound provides an explicit analytic bound for the quadratic form:
 
 {  
 \begin{equation}
   \scalar{\Phi}{-\Delta_{\Xi}\Phi}_{\Xi}\geq -
   \left[\kappa_m(\textstyle{\frac{\bar{\mu}}{1-\delta}})\right]^2   \scalar{\Phi}{\Phi}_{\Xi}\,,
   \label{nbound}
   	\end{equation}
	}
	i.e., $\mu_0=\,\kappa_m(\textstyle{\frac{\bar{\mu}}{1-\delta}})$.

	Let us try to give some further estimates. Assume that $\epsilon\kappa\gg1$\,. Then $\kappa\tanh(\epsilon\kappa)\approx\kappa$ and thus the spectral condition becomes $\kappa\approx c$ whenever $\epsilon\kappa\approx\epsilon c\gg 1$\,. In the case that $c=\bar{\mu}/ {(1-\delta)}$\,, this gives 
	$$\mu_0\equiv\kappa\approx \frac{\bar{\mu}}{{(1-\delta)}}{
	}$$ 
whenever $\epsilon\kappa\gg 1$\,, that is, whenever $\bar{\mu}/{{(1-\delta)}}\approx\mu_0\gg \frac{1}{\epsilon}$\,.

	Now assume that $\epsilon\kappa\ll 1$\,. In this case we have that $\kappa\tanh(\epsilon\kappa)\approx \epsilon\kappa^2$ and thus the solution of the spectral equation becomes $\kappa\approx \sqrt{c/\epsilon}$ whenever $\epsilon\kappa\approx \sqrt{\epsilon c}\ll 1$.  Thus, we have that 
		$$\mu_0\approx
	\sqrt{\frac{\bar{\mu}}{\epsilon{({1-\delta})}}}\quad \text{whenever}\quad \frac{\bar{\mu}}{{1-\delta}}\ll\frac{1}{\epsilon}\;.$$

\subsection{Existence of edge states with Robin boundary conditions} \label{sectionedge}

In the previous section  we have shown that for Robin boundary conditions the Laplace-Beltrami operator is always bounded from below. We will show now that  when the Robin parameter is positive, ${\mu}>0$\,, negative energy edge states appear for large positive values of $\mu$. This result has been proved in
 \cite{As05} in a more general framework so that it applies to our particular case. However the bounds that we are going to obtain are sharper. To prove the existence of such a negative energy eigenstate it will be enough to obtain a negative upper bound for the expectation value function of the Laplace-Beltrami operator for some function in an appropriate space.

From Green's identity it is straightforward to obtain the upper bound
	\begin{align}
	\scalar{\Phi}{-\Delta_M\Phi}_M= \scalar{\nabla\Phi}{\nabla\Phi}_M-\scalar{\varphi}{\mu\varphi}_{\partial M}\leq\scalar{\nabla\Phi}{\nabla\Phi}_M-\underline{\mu}\scalar{\varphi}{\varphi}_{\pM}, \label{cota}\
	\end{align}
 with $\Phi\in\D_{M}$ and where $\underline{\mu}:=\inf \mu(p)$. This implies that the quadratic form associated to the Laplace-Beltrami operator $-\Delta_M$
 with local Robin boundary condition given by $\mu(\bt)$ is bounded above by the form associated with the Robin boundary condition with constant parameter $\underline{\mu}$.
 
 We shall need the following slight generalisation of Davies' Theorem, \cite[Theorem 7.2.1]{Da95}. It states that smooth functions with any value of the normal derivatives at the boundary are dense in $\H^1(M)$\,, cf. \cite[Lemma~4.1]{ibort13}.
 
 \begin{lemma}\label{lem:approxdotphi}
Let $\Phi\in\H^1(M)$ and $f\in\H^{1}(\pM)$. Then, for every $\epsilon>0$ there exists $\tilde{\Phi}\in\mathcal{C}^\infty(M)$ such that
$\norm{\Phi-\tilde{\Phi}}_1<\epsilon$, $\norm{\varphi-\tilde{\varphi}}_{\H^{1/2}(\pM)}<\epsilon$ and
$\norm{f-\dot{\tilde{\varphi}}}_{\H^{1}(\pM)}<\epsilon$\,.
\end{lemma}

As in the proof of Theorem~\ref{thm:scalarthm}, it is enough to obtain the bound \eqref{cota} for  some function in $\Phi\in\overline{\D_M}^{\norm{\cdot}_{\H^1(M)}}$\,. The lemma above establishes that $\overline{\D_M}^{\norm{\cdot}_{\H^1(M)}}=\H^1(M)$\,. We can use again the splitting of the manifold showed in Fig.~\ref{fig:split M} and for functions $\Phi\in\H^1(\Xi)\oplus\H^1(M\backslash \Xi)$ we can split the quadratic form as 
 
 \begin{equation}
	\scalar{\Phi}{-\Delta_{_{M}}\Phi}_M =\scalar{\Phi}{-\Delta_{_{\Xi}}\Phi}_\Xi + \scalar{\Phi}{-\Delta_{_{M\backslash\Xi}}\Phi}_{M\backslash\Xi}\;.
	\label{semiboundXibb}
 \end{equation}

Now, consider functions on   $\phi\in\D_{\Xi}$ which only depend on the radial coordinate $r$,
$\phi(r)$, and extend them to the whole $M$ by defining 
\begin{equation}\label{eq:edgestateshape}
	\Phi_{\Xi}(x) =
		\begin{cases}
			\phi(r) &, x = (r, \bt) \in \Xi \\
			\phi(R_0-\epsilon) &, x \in M\backslash\Xi\\ 
		\end{cases}\;.
\end{equation}

We have that 
	\begin{align}
		\scalar{\Phi_{\Xi}}{-\Delta_{M}\Phi_{\Xi}}_{M}&=\scalar{\Phi_{\Xi}}{-\Delta_{\Xi}\Phi_{\Xi}}_{\Xi}\notag\\
		&\leq\int_{\partial M}\int_{R_0-\epsilon}^{R_0} \frac{\partial\bar{\Phi}_{\Xi}}{\partial r}\frac{\partial\Phi_{\Xi}}{\partial r}\d\mu_{g_{_{\Xi}}}-\underline{\mu}\scalar{\phi_{\Xi}}{\phi_{\Xi}}_{\partial M}\notag\\
								& \leq (1+\delta)\int_{\partial M}\left[\int_{R_0-\epsilon}^{R_0} \Bigl[ \frac{\partial\bar{\Phi}_{\Xi}}{\partial r}\frac{\partial\Phi_{\Xi}}{\partial r}\Bigr] \d r -\frac{\underline{\mu}}{1+\delta}|\phi_{\Xi}|^2 \right]\d\mu_{\tilde{g}}.\label{nadabb}
	\end{align}
	
Thus, since the quadratic forms at both sides are defined on the same domain, $\H^1(M)$, it is enough to find an upper bound of the one-dimensional problem \eqref{laplacian interval2} of the Laplacian on the interval $[0,\epsilon]$ with Robin boundary condition at $r=0$ and Neumann boundary condition  
	 at $r=\epsilon$.

As we have seen in the previous section, this one-dimensional operator has an edge state $\Phi_0$ which is given by the zero of the spectral function
\eref{nrbc} with Robin parameter $c={\underline{\mu}}/{(1+\delta)}$.
The solution of equation \eqref{nrbc}, $\kappa(c)$\,, is bounded from below\,, $\kappa_M(c)<\kappa(c)$,
by
{
  	\begin{equation}\label{ba}
	\kappa_M(c):=c \tanh \epsilon c.
	\end{equation}
	}
Thus if we take $\Phi_\Xi(r,\bt)=\Phi_0(r)$ for $(r,\bt)\in\Xi$ we get:
	\begin{equation}
		\scalar{\Phi_\Xi}{-\Delta_{_{M}}\Phi_\Xi}_M\leq -
		\left[\kappa_M(\textstyle{\frac{{\underline{\mu}}}{1+\delta}})\right]^2 \!  \scalar{\Phi_\Xi}{\Phi_\Xi}_{\Xi},\label{bound2}
   	\end{equation}
which shows the existence of negative modes of the Laplace-Beltrami operator. Notice that even though the function $\kappa_M(c)$ is positive, the eigenvalue associated to the boundary condition defined by $c$ is $-\kappa_M(c)^2$\,. These modes are edge states
and are localised at the boundary of $M$. Notice that this upper bound for edge states is lower than the bound found in Ref. \cite{As05}, although 
the latter one applies to more general boundary conditions. 
Every parameter, including the value of $\kappa_M(\textstyle{\frac{{\underline{\mu}}}{1+\delta}})$ depends only on the geometry of the neighbourhood of the boundary $\Xi$ (see e.g. \cite{PP15,kp15}).

In summary, combining the  two compatible bounds for edge states, we have the very constrained chain of inequalities
 \begin{equation}\nonumber
  -  \left[\kappa_m(\textstyle{\frac{{\bar{\mu}}}{1-\delta}})\right]^2  \!\! \scalar{\Phi_\Xi}{\!\Phi_\Xi}_M\,\leq
      \scalar{\Phi_\Xi}{\!\minus\Delta_{M}\Phi_\Xi}_M 
   \leq \! \minus\, 
\left[\kappa_M(\textstyle{\frac{{\underline{\mu}}}{1+\delta}})\right]^2 \!\!  \scalar{\Phi_\Xi}{\!\Phi_\Xi}_{\Xi}\,.
   	\end{equation}
However,  there is a significant difference between the two bounds. The upper bound shows the existence of edge states with
negative energy. This is so because it means that there is at least one eigenvalue with negative energy. However, the lower bound applies to any state and establishes a global lower bound which guarantees stability of
the vacuum for theories with mass larger than this bound.

In the asymptotic regime $\bar{\mu}>\underline{\mu}\gg\frac1{R_0}$\,, the bounds become
 \begin{equation}\label{eq:asyntotic1}
     -\textstyle{\frac{\bar{\mu}^2}{{(1-\delta)^2}}}  \scalar{\Phi_\Xi}{\Phi_\Xi}_M\,\leq
      \scalar{\Phi_\Xi}{-\Delta_{M}\Phi_\Xi}_M  
   \leq -\textstyle{\frac{{\underline{\mu}}^2}{{(1+\delta)^2}}} \scalar{\Phi_\Xi}{\Phi_\Xi}_{\Xi}.
   	\end{equation}
{	Since the inequalities hold for any value of $\delta\in (0,1)$ they become even more stringent,
 \begin{equation}\label{eq:asyntotic2}
     -\textstyle{{\bar{\mu}^2}{{}}}  \scalar{\Phi_\Xi}{\Phi_\Xi}_M\,\leq
      \scalar{\Phi_\Xi}{-\Delta_{M}\Phi_\Xi}_M  
   \leq -\textstyle{{{\underline{\mu}}^2}{{}}} \scalar{\Phi_\Xi}{\Phi_\Xi}_{\Xi},
   	\end{equation}
in agreement with the asymptotic  results of Refs. \cite{PP15,kp15}.
}
For small sizes  $\frac{1}{\epsilon}>\frac1{R_0}\gg\bar{\mu}>\underline{\mu}$\,, the bounds become
 \begin{equation}\label{eq:asyntotic3}
     -\textstyle{\frac{\bar{\mu}}{\epsilon{(1-\delta)}}}\scalar{\Phi_\Xi}{\Phi_\Xi}_M\,\leq  \scalar{\Phi_\Xi}{-\Delta_{M}\Phi_\Xi}_M   \leq -\textstyle{\frac{{\underline{\mu}}}{{\epsilon(1+\delta)}} } \scalar{\Phi_\Xi}{\Phi_\Xi}_{\Xi}\;.
   	\end{equation}
This state becomes a zero mode in the case of Neumann boundary conditions ($\mu=0$).
 
The above inequalities show that negative energy states appear as long as $\mu>0$\,. The value of the energy gets progressively more negative the larger is the value of the Robin parameter, or equivalently the smaller is the manifold. The splitting procedure used for the proof shows that the lowest eigenstate can be approximated in the direction normal to the boundary by the lowest eigenstate of the 1 dimensional Laplacian with Robin boundary conditions. The latter, as it is easy to check, gets more localised at the boundary as the value of $\mu$ increases.

In summary, we have proved the following theorem.
	
{\begin{theorem} Let $M$ be an oriented Riemannian manifold with regular oriented boundary $\partial M$. 
The Laplace-Beltrami operator $-\Delta$ restricted to smooth functions satisfying  the Robin boundary conditions 
\eref{bcp2}
has at least one negative eigenvalue. As   $\underline{\mu}=\inf \mu $  becomes large  the corresponding eigenfunction gets more localised on a small
collar $\Xi$ around  the boundary $\pM$.
\end{theorem}

	The local character of the Robin boundary condition allowed us to find a lower bound in the spectrum. However, there are more singular boundary conditions which do not have such a bound for any size of the physical
system \cite{As05,sachin}.

\section{Edge States of the electromagnetic field}\label{sectionEM}

In pure electrodynamics on the space-time $M\times {\mathbb R}$ the  dynamical fields 
are one-forms
$A\in\Lambda^1(M)$ over the space $M$. The corresponding potential induced by the magnetic term of Maxwell action is up to a surface term
\begin{equation}
	\mathcal{V}(A)=\frac{1}{2e^2}\int_M F\wedge \star F=\frac{1}{2e^2} \int_M \d A\wedge \star \d A= \frac{1}{2e^2} \scalar{A}{ \d^\ast \d A},
	\label{potential}
\end{equation}
where $\d^\ast$ is the codifferential  and $\scalar{\cdot}{\cdot}$ is the canonical scalar product of one-forms
\begin{equation}\label{eq:scalaroneforms}
	\scalar{\alpha}{\beta}=\int_M g^{-1}(\alpha,\beta) \d\mu_{g},\quad \alpha,\beta\in\Lambda^1(M)\;,
\end{equation}
induced by  the Riemannian metric  $g$ of $M$.  Here, $g^{-1}\scalarp{\cdot}{\cdot}$  denotes the canonical scalar product of  forms induced by $g$,  $\d\mu_g$ the  Riemannian volume on $M$ and  $\star:\Lambda^k(M)\to\Lambda^{n-k}(M)$ the Hodge star operator, i.e.  the unique operator that verifies
\begin{equation}\label{def:hodgestar}
	g^{-1}\scalarp{\alpha}{\beta}\d\mu_{g}=\alpha\wedge\star\beta,\quad\forall \alpha,\beta\in\Lambda^k(M)\;,
\end{equation}
where $n$ is the dimension of $M$\,.
The following identities will be needed later on, cf.~\cite{Wa71}. Let $\alpha,\beta\in \Lambda^k(M)$\,.
\begin{equation}\label{eq:scalarforms}
	\scalar{\alpha}{\beta} = \int_M g^{-1}\scalarp{\alpha}{\beta}\d\mu_{g}=\int_M \alpha\wedge\star\beta\;,\\
\end{equation}
\begin{equation}
	\scalar{\alpha}{\beta}=\scalar{\star\alpha}{\star\beta},\quad \star 1=\d\mu_{g},\quad \star\d\mu_{g}= 1,\quad 
	\star\star\alpha=(-1)^{k(n-k)}\alpha\;.
\end{equation}
In quantum electrodynamics the Hodge-de Rham Laplace operator 
\begin{equation}
-\Delta=\d^\ast \d
\label{dR}
\end{equation}
plays the same role as the Laplace-Beltrami operator in  scalar field theories. 

Using the Hodge star operator, the Laplace operator defined in Eq.~\eqref{dR} can be written, for $A\in\Lambda^1(M)$\,, as
\begin{equation}\label{def:laplaceoneform}
	\d^\ast\d A=(-1)^{n+1}\star\d\star\d A\;.
\end{equation}

Gauge invariance implies that the Hodge-de Rham-Laplace operator \eref{dR} has an infinite number of zero-modes of the form $A=\d\phi$. One can get rid of these spurious zero-modes by fixing the Coulomb gauge condition $\d^\ast A=0$ or using Feynmann gauge by adding a term of the form $\d \d^\ast$ to the operator \eref{def:laplaceoneform} 
to get the standard Hodge-de Rham Laplace operator
\begin{equation}
-\Delta_1=\d^\ast \d+ \d\, \d^\ast.
\label{HdR}
\end{equation}

In the rest of this section we shall prove that for a particular choice of Robin-like boundary conditions the non-elliptic operator $\d^\ast\d$ is bounded from below independently of the gauge fixing condition chosen.

\subsection{Boundary term for Laplace operator on one-forms}\label{subsec:BoundaryTerm}

As in the previous section, we will use the quadratic forms associated to the operators rather than the operators themselves to obtain the bounds. One important ingredient in the proof of Theorem~\ref{thm:scalarthm} was the use of Stokes Theorem to identify the boundary term associated to the Laplace operator. We shall first derive Green's identity for the Laplace operator \eqref{dR} and write the boundary term in an explicit form. 

 Using the identities above, one can get for $A\in\Lambda^1(M)$\,,
\begin{subequations}\label{subequations:green}
\begin{align}
	\scalar{A}{-\Delta A}&=(-1)^{n+1}\int_Mg^{-1}\scalarp{A}{\star\d\star\d A}\d\mu_{g}\label{subequations:greena}\\
		&=(-1)^{n+1}\int_M A\wedge\star\star\d\star\d A
		=\int_M A\wedge\d\star\d A\label{subequations:greenc}\\
		&=\int_M \d A \wedge \star \d A - \int_M \d\left( A \wedge\star\d A \right)\label{subequations:greend}\\
		&=\scalar{\d A}{\d A}-\int_{\pM}j^*(A \wedge\star\d A),\label{subequations:greene}
\end{align}
\end{subequations}
where $j:\partial M \hookrightarrow M$ is the canonical embedding of the boundary $\pM$
into the bulk manifold $M$ and $j^\ast$ denotes the corresponding pullback map of forms. 
In \eqref{subequations:greend} we have used the identity
$$\d( A \wedge\star\d A)=\d A \wedge\star\d A - A\wedge\d\star\d A$$
and in \eref{subequations:greene} the general Stokes' Theorem, cf. \cite[Theorem 8.2.8]{marsden01}.

Using  standard algebra of imbedding maps \cite{Du52}, the boundary term in \eref{subequations:greene} 
can be  rewritten as
	\begin{align}
		\int_{\pM}j^*(A\wedge\star\d A)&=\int_{\pM}j^*A\wedge\star_{\pM}j^*\mathrm{i}_{\mathbf{\cal{N}}}\d A  \notag\\
			&=\int_{\pM}\tilde{g}^{-1}\scalarp{j^*A}{j^*\mathrm{i}_{\mathbf{\cal{N}}}\d A}\d\mu_{\tilde{g}} \notag\\
			&=\int_{\pM}\tilde{g}^{-1}\scalarp{a}{\dot{a}}\d\mu_{\tilde{g}}=\scalar{a}{\dot{a}}_{\pM}\;,\label{eq:boundarytermof}
	\end{align}
where $\tilde{g}$ is the pull-back of $g$ to $\pM$, $\star_{\pM}$ denotes the Hodge star operator with respect to $\tilde{g}$,  {$\cal{N}$ the vector field defined in a neighborhood  of $\partial M$   by the geodesic flow induced by the outward normal vector $\mathbf n$ of the boundary $\partial M$,  and $\mathrm{i}_{\cal{N}}$ is the contraction with $\cal{N}$}. Also we have defined $a:=j^*A$ and $\dot{a}:=j^*\mathrm{i}_{\cal{N}}\d A$ maintaining the capital-small letter notation to distinguish objects on the manifold $M$ and its boundary $\pM$\,. Finally $\scalar{\cdot}{\cdot}_{{\pM}}$ stands for the canonical scalar product among one-forms at the boundary:
\begin{equation}\label{eq:scalaroneformsb}
	\scalar{\alpha}{\beta}_{\pM}=\int_{\pM} \tilde{g}^{-1}(\alpha,\beta) \d\mu_{\tilde{g}}=\int_{\pM} \alpha \wedge\star_{\pM}\beta\,.
\end{equation}
Compare the latter equation with Eq.~\eqref{eq:scalaroneforms}.

\subsection{Robin-like boundary conditions and semiboundedness}

In spite of the fact that the Laplace-de Rham operator $\d^\ast \d$ is gauge invariant like the original potential $\cal{V}$ in \eref{potential}, the boundary term 
\eref{eq:boundarytermof} is not manifestly gauge invariant. Thus, the most general boundary conditions break gauge invariance. However, there are
large families of boundary conditions which preserve gauge invariance, e.g. the Dirichlet boundary condition ($a=0,a_r=j^\ast \mathrm{i}_{\cal{N}} A=0$) or the Neumann
boundary condition $\dot a=0$.  One particularly interesting family preserving  the gauge symmetry is given by Robin-like boundary conditions defined by
\begin{subequations}\label{subequations:bconeforms}
\begin{align}
	 & j^\ast {\cal{L}}_{\cal{N}}  (A-\d (-\Delta_0)^{-1} \d^\ast A)= K j^\ast (A-\d(-\Delta_0)^{-1} \d^\ast A)\\
	&\ast_{\pM} j^\ast{\cal{L}}_{\cal{N}} 
	\ast(A-\d (-\Delta_0)^{-1} \d^\ast A)= 0,
\end{align}
\end{subequations}
where ${\cal{L}}_{\cal{N}}$ is the Lie derivative with respect to the geodesic vector field ${\cal{N}}$ defined in a neighborhood  of $\partial M$  by the geodesic flow induced by the normal vectors $\mathbf{n}$
of $\partial M$, $\Delta_0$ is the Laplace operator acting on scalars with Dirichlet boundary conditions and 
 $$K:\pM\to \mathbb{R}^{(n-1)\times (n-1)}\;$$ is a function with values in the space of ${(n-1)\times (n-1)}$ real orthogonal matrices (as required by the reality of the electromagnetic potential  $A$)  with respect to the pointwise scalar product among one-forms defined by $\tilde{g}^{-1}\scalarp{\cdot}{\cdot}$.

The invariance of the  boundary conditions \eref{subequations:bconeforms} under gauge transformations $ \phi$ which vanish at
the boundary $j^\ast \phi=0$ follows from the fact that in such a case
\begin{equation}\label{eq:gaugeinvariance}
	\d \phi-\d (-\Delta_0)^{-1} \d^\ast \d \phi=0.
\end{equation}

Notice that since $\Delta_0$ is the Laplace operator acting on scalars with Dirichlet boundary conditions, the boundary conditions \eref{subequations:bconeforms}
 reduce to 
 {
 \begin{subequations}
\label{subequations:bconeformsb}
\begin{align}
	 & j^\ast {\cal{L}}_{\cal{N}}  (A-\d (-\Delta_0)^{-1} \d^\ast A)(p)= K(p) j^\ast A(p), \label{subequations:bconeformsba}\\
	& \ast_{\pM} j^\ast{\cal{L}}_{\cal{N}}\ast (A-\d (-\Delta_0)^{-1} \d^\ast A)(p)= 0 \label{subequations:bconeformsbb}
\end{align}
\end{subequations}
for any $p\in\partial M$}.

The most relevant property of  these boundary conditions
is the cancellation of the   boundary term involved in the proof of self-adjointness  
of the Laplace-de Rham operator $\d^\ast \d$. Indeed, 
 such a boundary term  can be rewritten as
\begin{eqnarray}\label{eq:boundterm}
\!\!\!\!\!\!\!\!\!\!&&\!\!\!\!\!\!\!\!\!\!\!\!\!\!\scalar{A}{-\Delta \tilde{A}}-\scalar{\tilde{A}}{-\Delta A}=\scalar{\tilde{a}}{\dot{a}}_{\pM}-\scalar{a}{\dot{\tilde{a}}}_{\pM}\nonumber\\
&=&\int_{\pM} j^\ast \tilde{A}\wedge \ast_{\partial M}   j^\ast {\cal{L}}_{\cal{N}}  A- \int_{\pM} j^\ast \tilde{A}\wedge \ast_{\partial M}  j^\ast\d  \mathrm{i}_{\cal{N}}  A \nonumber\\
&-&\int_{\pM} j^\ast {A}\wedge \ast_{\partial M}   j^\ast {\cal{L}}_{\cal{N}}  \tilde{A}+ \int_{\pM} j^\ast {A}\wedge \ast_{\partial M}  j^\ast\d  \mathrm{i}_{\cal{N}} \tilde {A} \nonumber\\
\end{eqnarray}
{
or as
\begin{equation}\label{eq:boundterm2}
\begin{array}{r@{}l}
&{\phantom{-.\Biggr( } }\displaystyle \int_{\pM}   j^\ast \tilde{A}\wedge \ast_{\partial M} j^\ast {\cal{L}}_{\cal{N}}   (A-\d (-\Delta_0)^{-1} \d^\ast A) \\
&{\phantom{\Biggr( } }-\displaystyle \int_{\pM}j^\ast \tilde{A} \wedge \ast_{\partial M}  j^\ast\d  \mathrm{i}_{\cal{N}}   (A-\d (-\Delta_0)^{-1} \d^\ast A) \\
&{\phantom{\Biggr( } }-\displaystyle \int_{\pM}   j^\ast A \wedge \ast_{\partial M} j^\ast {\cal{L}}_{\cal{N}}  (\tilde{A}-\d (-\Delta_0)^{-1} \d^\ast \tilde{A}) \\
&{\phantom{\Biggr( } }+\displaystyle \int_{\pM} j^\ast A\wedge \ast_{\partial M} j^\ast \d   \mathrm{i}_{\cal{N}}    (\tilde{A}-\d (-\Delta_0)^{-1} \d^\ast \tilde{A}),
\end{array}
\end{equation}
since
$$\dot{a}=j^\ast {\cal{L}}_{\cal{N}}  A -  j^\ast\d  \mathrm{i}_{\cal{N}}  A= j^\ast {\cal{L}}_{\cal{N}}(A-\d (-\Delta_0)^{-1} \d^\ast A)-j^\ast\d  \mathrm{i}_{\cal{N}}(A-\d (-\Delta_0)^{-1} \d^\ast A).
$$
}
Because of the boundary condition \eref{subequations:bconeformsba}, the contributions of the first and third terms in \eref{eq:boundterm},  reduce to 
\begin{equation}\label{eq:bound2}
\scalar{a}{K {\tilde{a}}}_{\pM}-\scalar{K a}{\tilde{a}}_{\pM};
\end{equation}
which 
cancels out due to the orthogonal character of $K$. 
The remaining terms
{
\begin{equation}\label{eq:bound3}
 \int_{\pM}j^\ast A \,\wedge \ast_{\partial M}j^\ast \d   \mathrm{i}_{\cal{N}}  (\tilde{A}-\d (-\Delta_0)^{-1} \d^\ast \tilde{A})  -\!\int_{\pM} j^\ast \tilde{A}\, \wedge \ast_{\partial M}  j^\ast\d  \mathrm{i}_{\cal{N}}   (A-\d (-\Delta_0)^{-1} \d^\ast A) \nonumber
\end{equation}
can be rewritten as 
\begin{equation}\label{eq:bound4}
 -\int_{\pM} (\d^\ast_{_{\pM}} j^\ast A)\, j^\ast    \mathrm{i}_{\cal{N}}  (\tilde{A}-\d (-\Delta_0)^{-1} \d^\ast \tilde{A}) \,+\int_{\pM} (\d^\ast_{_{\pM}} j^\ast \tilde{A})\,  j^\ast  \mathrm{i}_{\cal{N}}  (A-\d (-\Delta_0)^{-1} \d^\ast A)  \nonumber
\end{equation}
}
by  integration by parts, where $\d^\ast_{_{\pM}}$ is the co-differential operator of $\pM$. Now, since
\begin{equation}\label{transverse}
\d^\ast  (A-\d (-\Delta_0)^{-1} \d^\ast A)= 0,
	\end{equation}
	we have that 
\begin{equation}\label{transverse2}
j^\ast\ast \d \ast  (A-\d (-\Delta_0)^{-1} \d^\ast A)=
{-}\, \d^\ast_{_{\pM}} j^\ast  A +\ast_{\pM} j^\ast{\cal{L}}_{\cal{N}}\ast (A- \d (-\Delta_0)^{-1} \d^\ast A)=0,
	\end{equation}	
	{
because $ \ast_{\pM}j^\ast A=
{ - \,}
 j^\ast  \mathrm{i}_{\cal{N}}\ast A$ and
	\begin{equation}\label{bulk-boundary}
j^\ast \ast \d \ast A= \ast_{\pM}j^\ast {\cal{L}}_{\cal{N}} \ast A\, {-}\, 
\d^\ast_{_{\pM}} j^\ast A.
	\end{equation}
	}
	The second term in \eref{transverse2} vanishes due to the boundary condition \eref{subequations:bconeformsbb}. Then,
	$$ \d^\ast_{_{\pM}}  j^\ast A=0,$$
which implies that the two boundary terms \eref{eq:bound3} also vanish. In summary, the boundary condition \eref{subequations:bconeformsb}
defines a  symmetric extension of the Laplace-de Rham operator.  It is easy to show that this extension is also essentially self-adjoint.

We shall show now that the self-adjoint extension of the Laplace-de Rham operator (Eq.~\eqref{dR}) acting on one-forms  constrained by these  boundary conditions  \eref{subequations:bconeformsb} 
 is bounded from below.

\begin{theorem}\label{thm:oneformthm} Let $M$ be an oriented Riemannian manifold with regular oriented boundary $\partial M$. 
Let $\mathbf n$ be the outgoing normal vector, $\cal{N}$ the vector field defined in a neighborhood  of $\partial M$   by the geodesic flow induced by $\mathbf n$
and $K$ a smooth function with values in ${(n-1)\times (n-1)}$ orthogonal 
matrices,  defined at the boundary $\pM$\,. 
The Laplace-de Rham operator $-\Delta_M$ restricted to smooth one-forms $A\in\Lambda^1(M)$ 
satisfying 
the Robin-like boundary conditions \eref{subequations:bconeformsb}
is essentially self-adjoint and bounded from below, i.e.
	\begin{equation}\label{semiboundoneform}
		\scalar{A}{-\Delta_M A}\geq-\mu_0^2\norm{A}^2 ,
	\end{equation}
where $\mu_0$ is a finite positive constant.
\end{theorem}

\begin{proof}
The proof will be almost identical to the proof of Theorem~\ref{thm:scalarthm} and will also be derived in several steps. Only Step~\ref{step:2} differs substantially. Again we shall use the splitting of the manifold into a collar neighborhood of the boundary $\Xi$ and its complement, cf. Fig.~\ref{fig:split M}.

\begin{enumeratesteps}

\item {\bf The collar Laplacian.}\label{step:1of}

Let $-\Delta_\Xi$ be the Laplace operator of Eq.~\eqref{def:laplaceoneform} on the manifold $\Xi$ defined on the domain 
{
	\begin{eqnarray}
	\D_\Xi= \left\{A  \in\H^2_{\Lambda^{1}}(\Xi) \right. &&| j^\ast {\cal{L}}_{\cal{N}} (A-\d (-\Delta_0)^{-1} \d^\ast A)(p)=0
 \,,p\in\partial\Xi_{-}\: ,\nonumber \\
 && j^\ast {\cal{L}}_{\cal{N}}   (A-\d (-\Delta_0)^{-1} \d^\ast A)(p)=K(p)a(p), p\in\pM\: ; \nonumber\\
	& &  \left.  \ast_{\pM} j^\ast{\cal{L}}_{\cal{N}}\ast (A-\d (-\Delta_0)^{-1} \d^\ast A)= 0,  p\in\partial\Xi\right\},
	\label{bcollar}
	\end{eqnarray}}
	where $\partial\Xi=\partial M \cup \partial {\Xi_-}$ and  $\H^2_{\Lambda^{1}}(\Xi)$ is the Hilbert space of one-forms (cf.\ Eq.\ \eqref{eq:scalarforms}) whose coefficients are in $\H^2(\Xi)$\,. This operator is essentially self-adjoint, as shown in the previous subsection.

\item {\bf Bound on the collar Laplacian.}\label{step:2of}
The collar Laplacian and therefore the associated quadratic form are bounded from below, i.e.,
$$\scalar{A}{-\Delta_{\Xi} A}\geq-\mu_0^2\norm{A}_\Xi^2\,,\quad A\in\overline{\D_{\Xi}}^{\norm{\cdot}_{\H^1_{\Lambda^1}(\Xi)}}\;,$$
where $\mu_0$ is of the order of $\sup|K(p)|$ and $|K(p)|$ stands for the norm of the matrix $K(p)$\,. Again the bar over the domain stands for the closure with respect to the corresponding norm. We postpone the proof of this step until the end of this section.

\item {\bf The inner Laplacian.}\label{step:3of}
 
Let $-\Delta_{M\backslash \Xi}$ be the Laplace operator on the manifold $M\backslash \Xi$ defined on the domain
{
\begin{equation}
\begin{array}{r@{}l}
\displaystyle\D_{M\backslash \Xi}=\left\{A\in\H^2_{\Lambda^{1}}\right. &{}\!({M\backslash \Xi}) | j^\ast {\cal{L}}_{\cal{N}} (A-\d (-\Delta_0)^{-1} \d^\ast A)(p)=0,\\
 &{}\displaystyle \left.  \ast_{\pM} j^\ast{\cal{L}}_{\cal{N}}\ast (A-\d (-\Delta_0)^{-1} \d^\ast A)=0\,,p\in\partial\Xi_{-}\right\}.
\end{array}
\end{equation}
}
This operator is positive and by \cite[Theorem 7.2.1]{Da95} its quadratic form has domain $\H^1_{\Lambda^{1}}(M\backslash\Xi)$ without any constraints at the boundary.

\item {\bf Equivalence of the bounds on the operator and on the associated quadratic form.}\label{step:4of}

By the same arguments used in Steps~\ref{step:4} and \ref{step:5} of the proof of Theorem~\ref{thm:scalarthm}, in order to find a bound for $-\Delta_M$ with domain 
{
\begin{eqnarray}\nonumber
\D_M\!=\!\left\{A\in\H^2_{\Lambda^{1}}(M)\right.&& |  j^\ast {\cal{L}}_{\cal{N}}  (A-\!\d (-\Delta_0)^{-1} \d^\ast A)=Kj^\ast A,\\
&&\left.  \ast_{\pM} j^\ast{\cal{L}}_{\cal{N}}\ast (A\!-\!\d (-\Delta_0)^{-1} \d^\ast A)= 0\right\}
\end{eqnarray}
}
it is enough to find the following bound for its quadratic form:
$$\scalar{A}{-\Delta_MA}_M\geq -\mu_0^2\norm{A}^2\,,\quad A\in\overline{\D_M}^{\norm{\cdot}_{\H^1_{\Lambda^1}(M)}}\;.$$
For $A\in\overline{\D_M}^{\norm{\cdot}_{\H^1_{\Lambda^1}(M)}}\subset \overline{\D_\Xi}^{\norm{\cdot}_{\H^1_{\Lambda^1}(\Xi)}}\oplus\H^1_{\Lambda^{1}}(M\backslash \Xi)$ we have that
\begin{align}
\scalar{A}{-\Delta_M A}_M&= \scalar{A}{-\Delta_{\Xi}A}_\Xi+\scalar{A}{-\Delta_{M\backslash\Xi}A}\notag\\
	&\geq -\mu^2_0\norm{A}^2_{\Xi}\geq -\mu^2_0\norm{A}^2_{M}\;,
\end{align}
where we have used \ref{step:2of} and \ref{step:3of}.
\end{enumeratesteps}

\end{proof}

\begin{proof2}[ of Step~\ref{step:2of}]
The proof will follow similar lines to those of Step~\ref{step:2} of the proof of the Theorem~\ref{thm:scalarthm}, but in this case we will bound the quadratic form $\scalar{A}{-\Delta_\Xi A}_\Xi$ by that of a one-dimensional problem with a flat metric.

First we remark that the operator $-\Delta_\Xi $ has an infinite dimensional kernel. More concretely, all the exact one-forms of ${\cal{D}}_\Xi$  \eref{bcollar} are in the kernel of the Laplacian. That is, if $A_\phi=\d \phi \in  {\cal{D}}_\Xi$, then $-\Delta A_\phi=0$. 
One particular class of such exact one-forms $A_\phi=\d\phi$ are those defined by fields $\phi$ 
vanishing at the boundary of $M$, i.e.
 $j_{\partial M}^*\phi=0$. The last condition guarantees that the zero modes satisfy the boundary conditions defined by \eref{bcollar} since ${a}_\phi=0$ and \eref{eq:gaugeinvariance} implies that the boundary conditions \eref{subequations:bconeformsb} become trivial identities.
 
{
Moreover, the gauge field defined by 
$$A':=A-\d\phi$$
satisfies the same boundary conditions as $A$ because of  the identity
\eref{eq:gaugeinvariance}.
}
 
\noindent
Because of this freedom we can chose $\phi$ in such a way that $A'$ is in Coulomb gauge $\d^\ast A'=0$.
This can be achieved for any gauge field $A\in \H^1_{\Lambda^{1}}(\Xi)$ by defining 
\begin{equation}
\phi=-\Delta^{-1}_0 \d^\ast A,
\end{equation} 
where $-\Delta_0$ is the self-adjoint Laplacian operator with Dirichlet boundary conditions densely defined on
$\L^2(\Xi)$.

The Coulomb gauge fixing  condition induces a decomposition of the Hilbert space $\L^2(\Lambda^1(\Xi))$ of  gauge fields  into two orthogonal complements
\begin{equation}
\L^2(\Lambda^1(\Xi))\cap\overline{\{A, \d^\ast A=0\}}\oplus \L^2(\Lambda^1(\Xi)) \cap\overline{\{A=\d\phi;\phi\in\H^2_{\Lambda^{1}}(\Xi), \phi_{_{\partial \Xi}}=0\}}.\label{decomposition}
\end{equation}
Now, for any  $A\in \D_\Xi$ we have
		\begin{equation}
-\Delta_\Xi A=-\Delta_\Xi A'=- \frac1{\sqrt{\det \widetilde{g}}}\partial_r\sqrt{\det \widetilde{g}}\, \partial_r  A' +\d^\ast_\theta \d_\theta A'  , 
	\label{proje}
	\end{equation}
	where $\d_\theta$ and $\d^\ast_\theta$ denote the exterior differential  and its adjoint in the Riemannian submanifolds 
	$r= cte.$ of $\Xi$, and where  we used  the decomposition of the Riemannian metric in terms of the Riemannian normal coordinates on the collar neighborhood:
$$\label{split inverse metric oneform}
	g^{-1}_{_{\Xi}}=
		\begin{pmatrix}
			1 & 0 \\ 0 & \tilde{g}^{-1}(r,\bt ) \\ 
		\end{pmatrix}\;.
$$
Notice that in this case the geodesic vector field ${\cal{N}}$ is given by the tangent vector to the radial geodesics, i.e. ${\cal{L}}_{\cal{N}}=\partial_r$ and the boundary conditions on the fields $A'$ become
{
	\begin{eqnarray}
	&& \partial_r A'(p)=0
 \,, \mathrm{for}  \, p\in\partial\Xi_{-}\: ,\nonumber \\
 && j^\ast \partial_r  A'(p)=K(p)a'(p), \: \partial_r\sqrt{\tilde{g}} A'_r (p) = 0, \mathrm{for}  \, p\in\pM.
	\label{bcollarf}
	\end{eqnarray}}
Thus, we have the following inequalities,	
	\begin{align}
	\!\!\!\!\!\!\!-\scalar{A}{\Delta_\Xi A}_\Xi\!&=\! -\scalar{A'}{\Delta_\Xi A'}_\Xi=-\scalar{A'\!}{  \frac1{\sqrt{\det \widetilde{g}}}\partial_r\sqrt{\det \widetilde{g}}\, \partial_r A'}_\Xi \! + \norm{\d_\theta A'}^2\notag \\ 
			&\!\geq\!\scalar{\partial_r A'}{\partial_r A'}_\Xi\!  - \!\scalar{a'\!}{K a'}_{\partial M} \geq\!\scalar{\partial_r A'}{\partial_r A'}_\Xi\!  - \bar{K}\scalar{a'\!}{a'}_{\partial M},
	\end{align}
where $\bar{K}=\mathrm{sup}\,  | \hspace{-2.3pt} |  K(p)|\hspace{-2.3pt} |$.
Let $\lambda_{\mathrm{max}}(p),\lambda_{\mathrm{min}}(p)$ be the highest and lowest eigenvalues of $g_\Xi^{-1}(p)$, respectively. 
Moreover, since the collar manifold $\Xi$ is compact, there is an infimum and a supremum of $\lambda_{\mathrm{max}}(p),\lambda_{\mathrm{min}}(p)$ in $\Xi$. The same argument holds for the induced metric at the boundary $\tilde{g}(R_0,\bt)$\,. The supremum and infimum of the eigenvalues will be denoted in this case by $\tilde{\lambda}_{\mathrm{max}}$ and $\tilde{\lambda}_{\mathrm{min}}$\,. 
We can now get the bound:
	\begin{align}
		\!\!\!\!\!-\scalar{A}{\Delta_\Xi A}_\Xi&\geq\lambda_{\mathrm{min}}\int_\Xi\sum_{i=1}^{n-1}{\partial_r A'_i}{\partial_r A'_i}\d\mu_g-\bar{K}\tilde{\lambda}_{\mathrm{max}}\int_{\partial M}\sum_{i=1}^{n-1} a'_ia'_i\d\mu_{\tilde{g}}\;\notag\\
		&\!\geq\lambda_{\mathrm{min}}(1-\delta)\!\!\int_{\partial M}\!\sum_{i=1}^{n-1}\Bigl[\int_I{\partial_r A'_i}{\partial_r A'_i}\d r\!-\!\frac{\bar{K}\tilde{\lambda}_{\mathrm{max}}}{(1-\delta)\lambda_{\mathrm{min}}}|a'_i|^2 \Bigr]\d\mu_{\tilde{g}}\;.
	\end{align}
The term in brackets corresponds for each $i=1,\dots,n-1$ to the one-dimensional problem analogous to the one appearing in \eqref{nada} and thus it is bounded from below by the same constant as that of the operator of Eq.\eqref{eq:radial operator}. Hence, 
\begin{align}
	-\scalar{A}{\Delta_\Xi A}_\Xi&\geq-\lambda_{\mathrm{min}}
	\sum_{i=1}^{{n}}\kappa\left( \textstyle{\frac{\bar{K}\tilde{\lambda}_{\mathrm{max}}}{\lambda_{\mathrm{min}}(1-\delta)}} \right)\norm{A'_i}_\Xi^2\notag\\
	&\geq -
	\kappa\left(\textstyle{\frac{\bar{K}\tilde{\lambda}_{\mathrm{max}}}{\lambda_{\mathrm{min}}(1-\delta)}}\right) \norm{A'}^2_\Xi =-C \norm{A'}^2_\Xi\; , \label{eq:finalboundof}
\end{align}
where  
\begin{equation}
C = 
\kappa\left(\textstyle{\frac{\bar{K}\tilde{\lambda}_{\mathrm{max}}}{\lambda_{\mathrm{min}}(1-\delta)}}\right).
\label{lb}
\end{equation}
and we have used that $\norm{A'}_\Xi^2\geq \lambda_{\mathrm{min}}\sum_i\norm{A'_i}^2_\Xi$.
Now, since $\norm{A}^2=\norm{A'}^2+\norm{\d \phi}^2$ and $C>0$ we also have that 
\begin{equation}
  \scalar{A}{-\Delta_\Xi A}_\Xi \geq - C\parallel A\parallel^2_\Xi.\label{gin}
\end{equation}
\medskip
i.e.  $-C$ is a lower bound of the operator $-\Delta_\Xi$
\end{proof2}
 
Due to the orthogonal decomposition \eref{decomposition} the lower bound \eref{lb} also holds for  the restriction $-\Delta$ to  the domain of  gauge fields which satisfy the Coulomb gauge condition
$$\d^\ast A'=0.$$
$-\Delta$ is also essentially self-adjoint in the domain of smooth gauge fields in the Coulomb gauge
which satisfy the Robin boundary conditions \eref{subequations:bconeformsb} that now become
{
\begin{equation}
	  j^\ast {\cal{L}}_{\cal{N}}  A= K j^\ast A \;,\quad
	  \ast_{\pM} j^\ast{\cal{L}}_{\cal{N}}\ast A= 0.
	\end{equation}}

 The only difference is a change in the degeneracy of eigenvalues and in the absence of trivial zero-modes.

The analysis of the Hodge-de Rham  operator $-\Delta_1$ of Feynmann gauge \eref{HdR} 
requires an extra boundary condition of Robin type $\mathrm{i}_{\mathbf{n}} A= \lambda\, d^\ast A$, where $\lambda$
is an arbitrary smooth function of the boundary $\partial M$. 

A similar analysis should lead to the existence of edge states and a lower bound
	\begin{equation}\label{semiboundhodge}
		\scalar{A}{-\Delta^1_M A}\geq-\mu_1^2\norm{A}^2,
	\end{equation}
	where the constant $\mu_1^2$ is also $\lambda$-dependent. An alternative proof could be obtained by using the  Weitzenb\"ock formula and a lower bound of the space curvature.

\section{Edge States for the Dirac Operator with Atiyah-Patodi-Singer boundary conditions} \label{sectionAPS}

The existence of edge states is not an exclusive property of scalar and vector fields. We shall show in this section that 
the Dirac operator with APS boundary conditions also has edge states.

Before we start the proof for the existence of edge states in this case, we will need to introduce some definitions. 
We assume that $M$ is a Riemannian spin manifold with  a given spin  structure. Let $\pi:E\to M$ be the vector bundle over M  associated to the fundamental representation of Spin(n) and $\omega$ the  spin connection
of $E$. The covariant derivative of $\omega$ in $E$ defines a map 
$$\nabla:\Gamma(E)\times\mathfrak{X}(M)\to\Gamma(E)\;,$$
where $\Gamma(E)$ stands for the space of sections of the vector bundle and $\mathfrak{X}(M)$ is the space of vector fields on $M$.
 }
 
In Riemann normal  coordinates, the Dirac operator acting on the restriction of the vector bundle $\pi:E\to M$  to
 the collar neighbourhood $\Xi=\pM\times\{R-\epsilon,R\}$ of $\pM$ \cite{Ba88}  reads
\begin{align}
\Dsl&=-i\slashed{\mathcal{N}}\cdot\nabla_{\mathcal{N}}-i\gamma_{\bt}\cdot\nabla_{\bt}+m\gamma^{d+1}\;,\\
	&=-i\slashed{\mathcal{N}}\nabla_{\mathcal{N}}+A'(m)
\end{align}
where $\slashed{\mathcal{N}}=\gamma^\mu {\mathcal{N}}_\mu$ with ${\mathcal{N}}_\mu$ the components of the vector field  ${\mathcal{N}}$ while $-i\gamma_{\bt}\cdot\nabla_{\bt}$ denotes the tangential component of $\slashed{D}$\,. We use the notation
\begin{equation}
	A'(m)=-i\gamma_{\bt}\nabla_{\bt}+m\gamma^{d+1}\;.
\end{equation}
Let the restriction to the boundary of this operator be
\begin{equation}\label{eq:operatorboundary}
A(m)=A'(m)|_{r=R}=-i\gamma_{\bt}\cdot\nabla_{\bt}+m\gamma^{d+1}\;.
\end{equation}

The matrix $\gamma^{d+1}$ is hermitean, anticommutes with all the $\gamma$'s and $(\gamma^{d+1})^2=\mathbb{I}$. The Dirac operator above is purely spatial and does not involve any time derivative.
$A(m)$ is an essentially self-adjoint operator acting on the Hilbert space of square integrable sections induced at the boundary that we denote by $\H(\pM)$. This is so because we assume that the boundary $\pM$ is a smooth compact manifold without boundary.

Let us define the  Atiyah-Patodi-Singer self-adjoint extensions of the Dirac operator $\Dsl$. Performing integration by parts, we get Green's identity for the Dirac operator, cf. \cite{booss93},
\begin{equation}\label{Green}
	\scalar{\Phi}{\Dsl\Psi}-\scalar{\Dsl\Phi}{\Psi}=i\scalar{\varphi}{\slashed{n}\psi}_{\pM}\;,
\end{equation}
where as usual we denote the restrictions to the boundary with small size Greek letters and  $\slashed{n}=\gamma^\mu n_\mu$ with $n_\mu$ the components of the normal vector field to $\pM$. The boundary conditions must force the R.H.S to vanish identically in order to define a symmetric extension. Let $\mathcal{K}$ be any self-adjoint operator on $\H(\pM)$ such that 
\begin{equation}\label{cond1}
	\mathcal{K}^2>0
\end{equation} 
and
\begin{equation}\label{cond2}
	\slashed{n} \mathcal{K}=-\mathcal{K}\slashed{n}\;.
\end{equation}
Since $\mathcal{K}^2$ is strictly positive by condition \eqref{cond1}, the Hilbert space at the boundary admits an orthogonal decomposition in terms of $\H^{\pm}(\pM)$, the positive and negative eigenspaces of $\mathcal{K}$ respectively.
The second condition, Eq. \eqref{cond2}, now guaranties that $\slashed{n}\H^+(\pM)=\H^-(\pM)$.

The generalised APS boundary conditions consist in requiring that $\Psi|_{\pM}\in\H^-(\pM)$ and therefore in defining the domain as
\begin{equation}\label{APS}
	\Df(\Dsl)=\left\{\Psi\in\H(E)\bigr|\Psi|_{\pM}\in\H^-(\pM)\right\}\;,
\end{equation}
where the Dirac operator is essentially self-adjoint, cf. \cite{At75,booss93}.\\

In order to analyse the existence of edge states we will need to analyse the relation between $\slashed{D}^2$ and another second order differential operator that one can define on the Riemannian vector bundle. The second covariant derivative is defined on a generic tensor field $T$ as 
$$\nabla^2_{X,Y}T:=\nabla_X\nabla_YT-\nabla_{\nabla_XY}T\;.$$

\begin{definition}\label{def:bochner}
	The connection Laplacian is the operator defined by 
	$$ -\Delta:=-\nabla^\mu\nabla_\mu\;.$$
\end{definition}

 There is an important relation between this Laplacian and $\slashed{D}^2$  known as  Lichnerowicz-Weitzenb\"ock identity, cf.\ \cite{lawson89}:
\begin{equation}\label{eq:Lichn-Weitz}
	\slashed{D}^2=-\Delta+R+m^2\;.
\end{equation}	
where $R$ is the Ricci scalar curvature of the Riemannian metric of $M$.

This theorem  defines a formal relation between the operators. Now we want to analyse the existence of edge states for the operator $\slashed{D}$ and this will be equivalent to analysing the existence of edge states for the operator $\slashed{D}^2$\,. But we need to select a domain for the latter that is compatible with the domain of $\slashed{D}$\,. We will define this domain by 
\begin{equation}\label{domaind2}
	\Df(\Dsl^2)=\left\{\Psi\in\Df(\Dsl)\bigr| \Dsl\Psi\in\Df(\Dsl)\right\}\;.
\end{equation}

With this choice, we have that the following identity holds for all $\Psi,\Phi\in\Df(\Dsl^2)$:
$$\scalar{\Phi}{\Dsl^2\Psi}=\scalar{\Dsl\Phi }{\Dsl\Psi}=\scalar{\Dsl^2\Phi}{\Psi}\;.$$
The boundary terms vanish by means of the conditions in \eqref{domaind2}. In fact, a general argument using the spectral theorem shows that $\mathfrak{D}(\slashed{D}^2)$ is an essentially self-adjoint domain for $\slashed{D}$\,.

We will now make a particular choice for $\mathcal{K}$ and hence of the boundary condition. We will show that this particular choice leads to low lying edge states. 

As the operator $\mathcal{K}$ for the generalised boundary conditions, we take the operator 
$$\mathcal{K}(\mu)\equiv i\slashed{n}A(\mu)\;.$$ 
It is clear that it verifies conditions \eqref{cond1} and \eqref{cond2}. If we denote by $P_+$ and $P_-$ the orthogonal projectors onto the boundary subspaces $\H^+(\pM)$ and $\H^-(\pM)$ respectively, the boundary condition on the domain of $\Dsl$ can be expressed as follows:
\begin{equation}\label{APS2}
	\Df(\Dsl):=\left\{ \Psi\in\H(M)\mid P_+\psi=0 \right\}\;.
\end{equation}

Let us write explicitly the boundary condition for $\Dsl^2$. We have that $\Psi\in\Df(\Dsl^2)$ implies that $\Dsl\Psi\in\Df(\Dsl)$\,.

Now
\begin{equation}
	\Dsl\Psi=-i\slashed{n}\nabla_{\mathcal{N}}\Psi+A'(m)\Psi
\end{equation}
\begin{equation}	
	(\Dsl\Psi)|_{\pM}=-i\slashed{n}\dot{\psi}+A(m)\psi\;,
\end{equation}
where $\dot{\psi}\equiv(\nabla_{\mathcal{N}}\Psi)|_{\partial M}$ and $\psi=\Psi|_{\pM}$\,. According to \eqref{domaind2} and \eqref{APS2} the domain of $\Dsl^2$ consists of the functions that fulfil $P_+\psi=0$ and $P_+(\Dsl\Psi)|_{\pM}=0$, which according to the identities above can be written as
\begin{align}
	&\phantom{aaaaaa}P_+\psi=0\notag\\
	&P_+\left[-i\slashed{n}\dot{\psi}+A(m)\psi\right]=0\;.\label{eq:boundaryDelta}
\end{align}

Recall that $P_+$ is the projection onto the positive (or more generally non-negative) space $\H^+(\pM)$ of the operator $\mathcal{K}(\mu)$\,. Our aim is to show that there exist low lying edge states. In the spirit of the proof on the existence of edge states for the Laplace operator given in \cite{As05}, we are going to select an edge state satisfying the boundary condition above. In order to do so, we need to analyse first the relation between the operators $\mathcal{K}(\mu)$ and $A(m)$\,. First notice that 
\begin{equation}
	A(m)=A(\mu)+(m-\mu)\gamma^{d+1}\;,
\end{equation}
\begin{equation}
	A(\mu)=(\slashed{n})^2A(\mu)=-i\slashed{n}\mathcal{K}(\mu)\;.
\end{equation}
Hence the boundary condition \eqref{eq:boundaryDelta} can be rewritten as
\begin{equation}\label{eq:BC}
	P_+\left[ -i\slashed{n}\dot{\psi}-i\slashed{n}\mathcal{K}(\mu)\psi+(m-\mu)\gamma^{d+1}\psi \right]=0\;.
\end{equation}
According to this, the most simple case occurs for $m=\mu$\,. Therefore, for simplicity, consider from now on that $\mu=m$ and that the boundary condition is
\begin{align}
	&\phantom{aaa}P_+\mathcal{K}(m)\psi=0\notag\\
	&P_+\left[ -i\slashed{n}\dot{\psi}-i\slashed{n}\mathcal{K}(m)\psi\right]=0\;.\label{eq:BC2}
\end{align}
Moreover, there is a direct relation between the square of the operator $\mathcal{K}(m)$ and the connection Laplacian on the boundary of the manifold given by the Lichnerowicz-Weitzenb\"ock identity, see Definition~\ref{def:bochner} and Eq.~\eqref{eq:Lichn-Weitz}:
\begin{equation}\label{eq:K2}
	\mathcal{K}(m)^2=i\slashed{n}A(m)i\slashed{n}A(m)=(\slashed{n})^2 A(m)^2=A(m)^2=-\Delta_{\bt}+m^2+{R}_{\pM}\;.
\end{equation}
Now ${R}_{\pM}$ stands for the scalar curvature of the induced Riemannian connection at the boundary.

In contrast to the situations on Sections~\ref{sectionRBC} and \ref{sectionEM}, the operator $\Dsl^2$ is automatically positive and thus there is no need to prove its semiboundedness. But we will show that there are edge states, whose eigenvalues are close to zero and therefore below the mass gap, which is not obvious. We will follow the approach taken in Section~\ref{sectionedge} and find an upper bound for the energy of certain states that are strongly localised at the boundary.

First we select an eigenstate $\xi\in\H_-$ of $\mathcal{K}(m)$ with eigenvalue $-\Lambda$, with $\Lambda>0$\,,
\begin{equation}\label{eq:eigenvalueK}
	\mathcal{K}(m)\xi=-\Lambda\xi\;.
\end{equation}
For this state we can prove the following bound.

\begin{lemma}\label{lem:boundxi}
	Let $\xi$ be an eigenstate of $\mathcal{K}(m)$ with eigenvalue $-\Lambda$ and let $-\Delta_{\bt}$ be the boundary Laplacian of the spin vector bundle induced at the boundary $\pM$\,. Then
	$$\scalar{\xi}{-\Delta_{\bt}\xi}\leq (\Lambda^2-m^2)\norm{\xi}^2+\sigma_{\mathrm{max}}\norm{\xi}^2\;,$$
where $\sigma_{\mathrm{max}}:=\sup_{p\in \pM}|{R}_{\pM}(p)|$\,.
\end{lemma}

\begin{proof}

From \eqref{eq:K2} and condition \eqref{eq:eigenvalueK}, we have that 
\begin{align}
	\scalar{\xi}{(-\Delta_{\bt}+m^2+\mathcal{R}_{\pM})\xi}_{\pM}&=\Lambda^2\|\xi\|^2\notag\\
	\Leftrightarrow\quad \scalar{\xi}{-\Delta_{\bt}\xi}_{\pM}&=(\Lambda^2-m^2)\|\xi\|^2-\scalar{\xi}{{R}_{\pM}\xi}_{\pM}\notag\\
	\Rightarrow\quad \scalar{\xi}{-\Delta_{\bt}\xi}_{\pM}&\leq(\Lambda^2-m^2)\|\xi\|^2+\sigma_{\mathrm{max}}\|\xi\|^2\;,
\label{boundboundarysigma}
\end{align}\\

\end{proof}

Notice that we can choose $\Lambda$ such that $\Lambda^2-m^2+\sigma_{\mathrm{max}}$ is the closest element to the smallest eigenvalue of $-\Delta_{\bt}$. 
It is well known that the eigenvalues of the connection Laplacian, cf.\ Definition~\ref{def:bochner}, on a manifold without boundary decrease with increasing volume of the manifold. As an example consider the situation for the sphere $S^n$ with radius $\rho$. The eigenvalues of the Dirac operator are given by $\lambda_l= \pm(l+n/2)/\rho$ , cf.\ \cite{camporesi96}. Using the Lichnerowicz-Weitzenb\"ock identity and the fact that in this case the scalar curvature is  $R = n(n-1)/(4\rho^2)$ one gets that the lowest eigenvalue for the connection Laplacian is $n/(4\rho^2)$. Hence, by increasing the volume of the manifold it is possible to make the smallest eigenvalue different from zero  arbitrarily small.\\

Now we will show (see Theorem~\ref{thm:edgedirac}) that if $\Lambda^2-m^2+\sigma_{\mathrm{max}}$ is small enough and $m$ is large enough, for a given width $\epsilon$ of the collar neighbourhood $\Xi$, there will appear negative energy states for $-\Delta$, which means that there will appear edge states for $\Dsl^2$ with energies close to zero. Since 
$$\scalar{\Dsl\Psi}{\Dsl\Psi}=\scalar{\Psi}{\Dsl^2\Psi}$$
for $\Psi$ fulfilling \eqref{domaind2} there are corresponding low-lying edge states. In practice these conditions will be met, as stated in the previous paragraph, for large enough manifolds.

Consider the following vector state with $\xi$ taken as in Lemma~\ref{lem:boundxi}:
\begin{equation}\label{edge state}
	\Psi=\exp\left( k \tan (\frac{\pi}{2\epsilon}(R_0-r)) \right)\xi(\bt)\;,
\end{equation}
for $R_0- \epsilon \leq  r \leq R_0$ and $\Psi \equiv 0$  if $ r \leq R_0 -\epsilon$.
Since $\Psi|_{r=R_0}=\psi=\xi$\,, it is clear that $\Psi\in\Df(\Dsl)$\,. Now we need it to be in the domain of $\Dsl^2$ and therefore we need to select $k$ in order that the boundary condition \eqref{eq:BC2} is satisfied. It is enough that 
\begin{equation}
	\dot{\psi}=-\mathcal{K}(m)\psi\;.
\end{equation}
A short calculation shows that for $\Psi$ as in Eq.~\eqref{edge state}, the left hand side is $\dot{\psi}=-\frac{k\pi}{2\epsilon}\xi$ while, according to \eqref{eq:eigenvalueK}, the right hand side verifies $-\mathcal{K}(m)\psi=\Lambda\xi$\,. Hence $\Psi\in\Df(\Dsl^2)$ if $$k=-\frac{2\epsilon\Lambda}{\pi}\,.$$ Moreover, under the latter condition the state $\Psi$ satisfies 
\begin{equation}\label{eq:robin dirac}
\dot{\psi}=\Lambda\psi\;.
\end{equation}

Let us now compute the mean value of $\Dsl^2$ in this vector state. We can assume that on the collar neighbourhood, we work on a frame of the vector bundle $E$ such that $\nabla_{\mathcal{N}}\Psi=\partial_r\Psi$. Then, 
\begin{align}
\scalar{\Psi}{-\Delta\Psi}&=\scalar{\nabla\Psi}{\nabla\Psi}-\Lambda\scalar{\xi}{\xi}_{\pM}=\int_{R_0-\epsilon}^{R_0}\int_{\pM}\d \mu_g \left|\frac{\partial \Psi}{\partial r}\right|^2\notag\\
	&\phantom{aaaa}+\int_{R_0-\epsilon}^{R_0}\int_{\pM}\d\mu_{g} \exp\left(2k\tan\frac{\pi}{2\epsilon}(R_0-r)\right)\xi^\dagger(-\Delta_{\bt}\xi)-\Lambda||\xi||^2_{\pM\notag}\\
	&\leq (1+\delta)||\xi||^2_{\pM}\biggl[ (\Lambda^2-m^2+\sigma_{\mathrm{max}})\int_{R_0-\epsilon}^{R_0}\exp\left(2k\tan\frac{\pi}{2\epsilon}(R_0-r)\right)\d r\notag\\
	 &\phantom{aaaa}+\int_{R_0-\epsilon}^{R_0}\Lambda^2(1+\tan^2(\frac{\pi}{2 \epsilon}(R_0-r)))^2\exp\left(2k\tan\frac{\pi}{2\epsilon}(R_0-r)\right)\d r-\frac{\Lambda}{1+\delta}\biggr]
	 \notag\\
	&\leq(1+\delta)||\xi||^2_{\pM}\left[ \Lambda\left(\frac{1}{2}+\frac{\pi^2}{16\epsilon^2\Lambda^2}-\frac{1}{1+\delta}\right)+(\Lambda^2-m^2+\sigma_{\mathrm{max}})\epsilon \right].\notag\label{eq:inequality}
\end{align}
Here we have used the explicit form of the function $\Psi$, Eq.~\eqref{edge state}, Eq.~\eqref{eq:robin dirac}, Lemma~\ref{lem:boundxi}, the identity 
\begin{equation}
	\int_0^{\pi/2}\d t (1+\tan^2 t)^2 \exp(-2k\tan t)=\frac{1}{2k}+\frac{1}{4k^3},
\end{equation}
and $\delta$ is defined in Eq.~\eqref{compact omega}. Thus we have proved the following:

\begin{theorem}\label{thm:edgedirac}
	Let  $-\Delta$ be the connection Laplacian, see Definition~\ref{def:bochner}, defined on the Hermitean bundle $\pi:E\to M$ and let $\Psi\in\Df(\Dsl^2)$ be a state of the form shown in Eq.~\eqref{edge state}, with $\xi(\bt)$ satisfying Eq.~\eqref{eq:eigenvalueK} and where $k=\textstyle{-\frac{2\epsilon\Lambda}{\pi}}$. Then the following bound holds.
	$$\scalar{\Psi}{-\Delta\Psi}\leq(1+\delta)||\xi||^2_{\pM}\left[ \Lambda\left(\frac{1}{2}+\frac{\pi^2}{16\epsilon^2\Lambda^2}-\frac{1}{1+\delta}\right)+(\Lambda^2-m^2+\sigma_{\mathrm{max}})\epsilon \right]\;.$$
\end{theorem}

Hence, as stated above, $\Lambda^2$ can be chosen of the order of $m^2-\sigma_{\mathrm{max}}$ when the mass is large enough and if the lowest eigenvalue of $-\Delta_{\bt}$ is small enough. In this case, the first term in the brackets of the last inequality above becomes negative, showing that $-\Delta$ has negative eigenvalues. Moreover, for large $\Lambda$ they become edge localised as follows from  the boundary condition $\dot{\psi}=\Lambda\psi$\,. In general, states with higher angular momentum might have energies greater than $m^2$ even though they are edge localised.

Bulk states can be defined as elements of $\Gamma_c(E)$, i.e., sections with compact support in the interior of the base manifold. Because for these sections $\Psi|_{\pM}=0$\,, one has that 
$$\scalar{\Psi}{-\Delta\Psi}\geq0\;,$$
and therefore $$ \scalar{\Psi}{\Dsl^2\Psi}\geq m^2\norm{\Psi}^2+\scalar{\Psi}{R\Psi}\;.$$

In contrast to the situation for scalar fields, there might be no eigenvalues of $\slashed{D}$ close to zero if the smallest eigenvalue of the boundary Laplacian $-\Delta_{\bt}$ is not small enough. Since this value depends on the volume of the boundary manifold, there will be a threshold value for the volume, which can be determined using Weyl's law, cf.\ \cite{lieb97}, such that edge states always exist. 

The difficulty of finding a negative upper bound for the expectation value of $-\Delta$ in the state \eref{edge state} suggests the existence of a critical size threshold 
for the existence of edge states. We shall show in the next section that this occurs in some cases. However, to get a general result we will need a general lower bound on  $-\Delta$ with the required transition properties.  

On the other hand, from the phenomenology of  macroscopic devices  the existence of such a threshold for edge states is expected. For instance, for small enough samples of topological insulators the effects of low-lying edge states are not present and a threshold for their appearance is detected experimentally \cite{2D}.

\section{Edge states  of Dirac operator with chiral bag boundary conditions} \label{sectionchiral}

In addition to the former, there are other families of boundary conditions that one can consider, for instance, the boundary conditions used in the analysis of quark confinement \cite{mit} or their generalizations like the chiral bag boundary conditions, cf.~\cite{chiral}. These can be introduced using a general approach that leads to self-adjoint Dirac Hamiltonians.

If the boundary term \eqref{Green} is written as the difference of the two chiral components  $\Psi_\pm=\frac12(1\pm\nsl)\Psi$  of spinors, $\Psi=\Psi_+ +\Psi_-$\,, then we get
\begin{equation}
	\Sigma(\Psi,\Phi)=i \scalarb{ \Psi_+}{ \Phi_+}- i\scalarb{\Psi_-}{ \Phi_-}.
\end{equation}
It is easy to check that boundary conditions of the form
\begin{equation}\label{cdos}
	(1-\nsl)\psi=U \gamma^{d+1} (1+\nsl)\psi\;,
\end{equation}
 where $U$ is a unitary operator on the boundary Hilbert space of spinors commuting with $\nsl$\,, make the boundary term vanish. These boundary conditions define domains of essential self-adjointness provided that certain regularity conditions on the unitary operator $U$ are satisfied, cf.\ \cite{As05,bruning08,ibort13}.
  
 We will assume ``local'' boundary conditions imposing that $U$ is a finite dimensional matrix acting only on spinor indices. Using the Cayley transform, we can express \eqref{cdos} as
\begin{equation}\label{dbc}
	\nsl \psi= \frac{1-U \gamma^{d+1}}{I+U\gamma^{d+1}} \psi\;.
\end{equation}
As a particular case, we can consider $U$ of the form
\begin{equation}\label{ubc}
	U=e^{2 i\arctan e^{\alpha}},
\end{equation}
which because of the identities
\begin{eqnarray*}\label{ibc}
	\!\!\!\!\frac{1-U \gamma^{d+1}}{I+U\gamma^{d+1}}\!\!&=&\!\!\frac{1+U }{I-U}(1-\gamma^{d+1})+\frac{1-U }{I+U}(1+\gamma^{d+1})\\
	\!\!&=&\!\!i\cot(\arctan e^{\alpha}) (1-\gamma^{d+1})-i \tan(\arctan e^{\alpha}) (1+\gamma^{d+1})\\
	\!\! &=&\!\!i e^{-\alpha} (1-\gamma^{d+1})-i e^{\alpha} (1+\gamma^{d+1})= -i e^{\alpha\gamma^{d+1}}\gamma^{d+1},
\end{eqnarray*}
corresponds  to the chiral bag boundary conditions \cite{chiral,jenalaplata}:
\begin{equation}\label{cbc}
	\frac12\left(1-i \gamma^{d+1}e^{-\alpha\gamma^{d+1}}\nsl\right)  \psi=0\;.
\end{equation}

 \subsection*{Example: The  3D-Ball $B_3$}

Let us consider a Dirac electron moving on a 3D ball $B_3$ of radius $R_0$. 
The Dirac Hamiltonian 
\begin{equation}
{
	H=-i\vec\gamma\cdot \vec{\nabla}+ \gamma_{0}m \;}
\end{equation}
 is subject to the boundary condition \eqref{ibc}:
\begin{equation}
	\nsl \Psi(R_0,\theta,\varphi)= -i  {\rm e}^{\alpha \gamma_0} \gamma_0
		 \Psi(R_0,\theta,\varphi)
	\label{rbcD}
\end{equation}
with $\nsl=\gamma_1 \sin\theta\cos\varphi+\gamma_2 \sin\theta\sin\varphi+ \gamma_3 \cos\theta$. 

The Hamiltonian  $H$ is essentially self-adjoint in the space of smooth functions satisfying (\ref{ibc}).
Let us consider the spectrum of  $H$. It is invariant under rotations
with generators given by the total angular momentum ${{J^i}}={{L^i}}+S^i$ where $S^i=-\frac{i}4\epsilon^{ijk}{\{\gamma_j,\gamma_k\}}$ (i=1,2,3). Let us
consider stationary states of the form (see e.g. \cite{Greiner,Strange,Thaller})
\begin{equation}\label{ansatz}
\Psi_{jm}(r ,\theta,\varphi)=		\begin{pmatrix}
			  \Omega_{j,m,\kappa}(\theta,\varphi) \phi_1(r)  \\ i\,  \Omega_{j,m,-\kappa}(\theta,\varphi) \phi_2(r) \\ 
		\end{pmatrix}
\end{equation}
where, for $\kappa>0$\,, 
$$ \Omega_{j,m,\kappa}(\theta,\varphi)=\frac1{\sqrt{2j}}\begin{pmatrix}\sqrt{j+m}\  Y_{j-\frac12}^{m_j-\frac12} \\
\sqrt{j-m} \ Y_{j-\frac12}^{m_j+\frac12}
\end{pmatrix}
$$ 
and 
$$ \Omega_{j,m,-\kappa}(\theta,\varphi)=\frac1{\sqrt{2j+2}}\begin{pmatrix}\sqrt{j-m+1}\  Y_{j+\frac12}^{m_j-\frac12} \\
-\sqrt{j+m+1} \ Y_{j+\frac12}^{m_j+\frac12}
\end{pmatrix}
$$ 
are the bispinor eigenfunctions of the three commuting operators $J^2, J_3, K=\gamma_0(2 S^i L_i+1)$
associated to angular momenta,
 \begin{eqnarray*}
J^2\Psi_{jm}(r ,\theta,\varphi)&=	&j(j+1)\, \Psi_{jm}(r ,\theta,\varphi)\\
J_3\Psi_{jm}(r ,\theta,\varphi)&=&	 m \, \Psi_{jm}(r ,\theta,\varphi)\\
K \Psi_{jm}(r ,\theta,\varphi)&=	&\kappa\,  \Psi_{jm}(r ,\theta,\varphi)= \left(j+{\textstyle{\frac12}}\right) \Psi_{jm}(r ,\theta,\varphi).\\
 \end{eqnarray*}
 Now, since
 \begin{equation}
-i\vec\gamma\cdot \vec{\nabla} \Psi_{jm}(r ,\theta,\varphi)= -i\vec\gamma\cdot \vec{n}\left(\partial_r +\frac1{r}( 1-\gamma_0 K)
\right) \Psi_{jm}(r ,\theta,\varphi),
\end{equation}
the solutions of the Dirac equation
 $$
 H \Psi_{jm}(r ,\theta,\varphi)= E \Psi_{jm}(r ,\theta,\varphi)
 $$
 must satisfy the coupled equations:
\begin{equation}\label{eq1}
\left(\partial_r +\frac{1-\kappa}{r}\right)\phi_1(r)- (E+m) \phi_2(r)=0,
\end{equation}
\begin{equation} \label{eq2}
			   \left(-\partial_r -\frac{1+\kappa}{r}\right) \phi_2(r)+(-E+m)\phi_1(r)=0,  
\end{equation}
which can be decoupled into a pair of second order differential equations:
\begin{equation}\label{eqq1}
			\left(\partial_r^2+\frac2{r} \partial_r - \frac{\kappa(\kappa-1)}{r^2}\right)\phi_1(r)= (m^2-E^2)\phi_1(r),
\end{equation}
\begin{equation}
			\left(\partial_r^2+\frac2{r} \partial_r - \frac{\kappa(\kappa+1)}{r^2}\right)\phi_2(r)= (m^2-E^2)\phi_2(r).
\end{equation}

Since $\{\gamma\cdot \vec{n},K\}=0$ and
 \begin{equation} \label{eqq2}
\vec\gamma\cdot \vec{n}\, \Psi_{jm}(r ,\theta,\varphi)= \begin{pmatrix}
			-i\, \Omega_{j,m,\kappa}(\theta,\varphi) \phi_2(r)  \\  - \Omega_{j,m,-\kappa}(\theta,\varphi) \phi_1(r) \\ 
		\end{pmatrix} ,
\end{equation}
the boundary conditions (\ref{rbcD}) imply that
\begin{equation}\label{rbc2}
			\phi_2(R_0)={\rm e}^\alpha \phi_1(R_0).
\end{equation}

The solutions of \eref{eqq1} and \eref{eqq2} with smooth behaviour at $r=0$ are given in terms of the 
Bessel functions \cite{Abramowitz, nistdl},

\begin{figure}[h]
 \centerline{  \includegraphics[height=8.5cm]{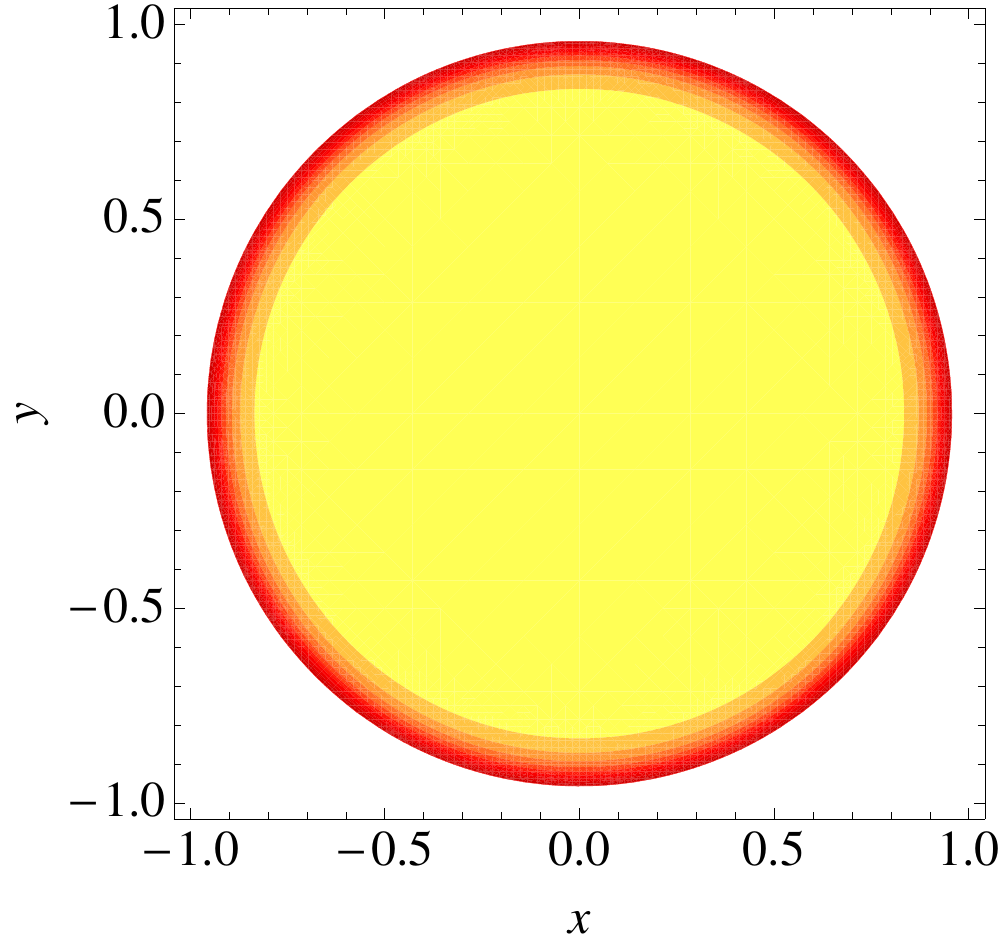}}
 \caption{{Charge density distribution $\Psi^\dagger \Psi$ of edge state with zero energy on a 3D ball $B_3$ for $\alpha=1.16144$ and $mR_0=1$.}}
  \end{figure}

\begin{equation} \label{hspin}
\phi_1(r)= \frac{\sqrt{m+E}}{\sqrt{r}}I_{\kappa-\frac12}\left(r \sqrt{m^2-E^2}\right),\ 	\phi_2(r)= \frac{\sqrt{m-E}}{\sqrt{r}} I_{\kappa+\frac12}\left(r \sqrt{m^2-E^2}\right), 
\end{equation}
where the normalization was chosen to match the Dirac equations \eref{eq1} and \eref{eq2}. Now, the boundary conditions \eref{rbcD}
imply that 
\begin{equation}
\label{spectrum}
			\sqrt{\frac{m+E}{m-E}}\frac{I_{\kappa-\frac12}\left(R_0 \sqrt{m^2-E^2}\right) }{I_{\kappa+\frac12}\left(R_0 \sqrt{m^2-E^2}\right) }= {\rm e}^{-\alpha},
\end{equation}
which is the spectral condition. The edge states correspond to   states with negative kinetic energies ($T=|E|-m$), i.e. $-m<E<m$ and
the simplest one  corresponds to $j=1/2$ and $\kappa=1$, i.e.
\begin{equation}
\phi_1(r)= \frac{\sqrt{m+E}}{\sqrt{r}}I_{\frac12}\left(r \sqrt{m^2-E^2}\right),\ 	\phi_2(r)= \frac{\sqrt{m-E}}{\sqrt{r}} I_{\frac32}\left(r \sqrt{m^2-E^2}\right). 
\end{equation}
These states always exist provided that $\alpha>0$.
 
Notice that in the limit $\alpha\to \infty$, (\ref{rbc2}) leads to Dirichlet boundary conditions for $\phi_1$, i.e. $\phi_1(R_0)=0$  whereas  in the limit $\alpha\to -\infty$ we obtain Dirichlet boundary condition for $\phi_2$, i.e $\phi_2(R_0)=0$ and because of
\eqref{eq2}, Neumann  boundary condition for $\phi_1$ when $j=\frac12$, i.e. $\phi'_1(R_0)=0$.

In particular, we have a zero mode for $\alpha = \log I_\frac32(R_0 m) -  \log I_\frac12(R_0 m) $ which corresponds to  the maximally localised edge state.
The concentration of the  state  on the edge increases as the mass gap increases, which provides the perfect situation for a topological insulator.

  One can find more states with higher angular momenta and negative kinetic energies.
  And whenever  $\alpha = \log I_{j+\frac12}(R_0 m) -  \log I_{j-\frac12}(R_0 m)$\,, we have  zero modes.

Notice that because the boundary condition \eref{rbcD} preserves rotational invariance, the energy spectrum is degenerate. For angular momentum $j$\,,
there are $2j+1$ degenerate states. There is an additional symmetry of Dirac equation in $\mathbb{R}^3$. If we interchange the two bi-spinors in the ansatz 
\begin{equation}\label{ansatz2}
\Psi_{jm}(r ,\theta,\varphi)=		\begin{pmatrix}
		\Omega_{j,m,-\kappa}(\theta,\varphi) \phi_1(r)	 \\ i\,  \Omega_{j,m,\kappa}(\theta,\varphi) \phi_2(r) \\ 
		\end{pmatrix}\;,
\end{equation}
it is possible to find solutions with the same energy. However, in the ball $B_3$ the boundary condition breaks this symmetry, and in particular  
there are no edge states of the form \eref{ansatz2} when $\alpha<0$\,.
This is in agreement with the experimental result that for small enough samples of topological insulators the effects of low-lying edge states are not present \cite{2D}.

    \begin{figure}[h]
 \centerline{  \includegraphics[height=5.5cm]{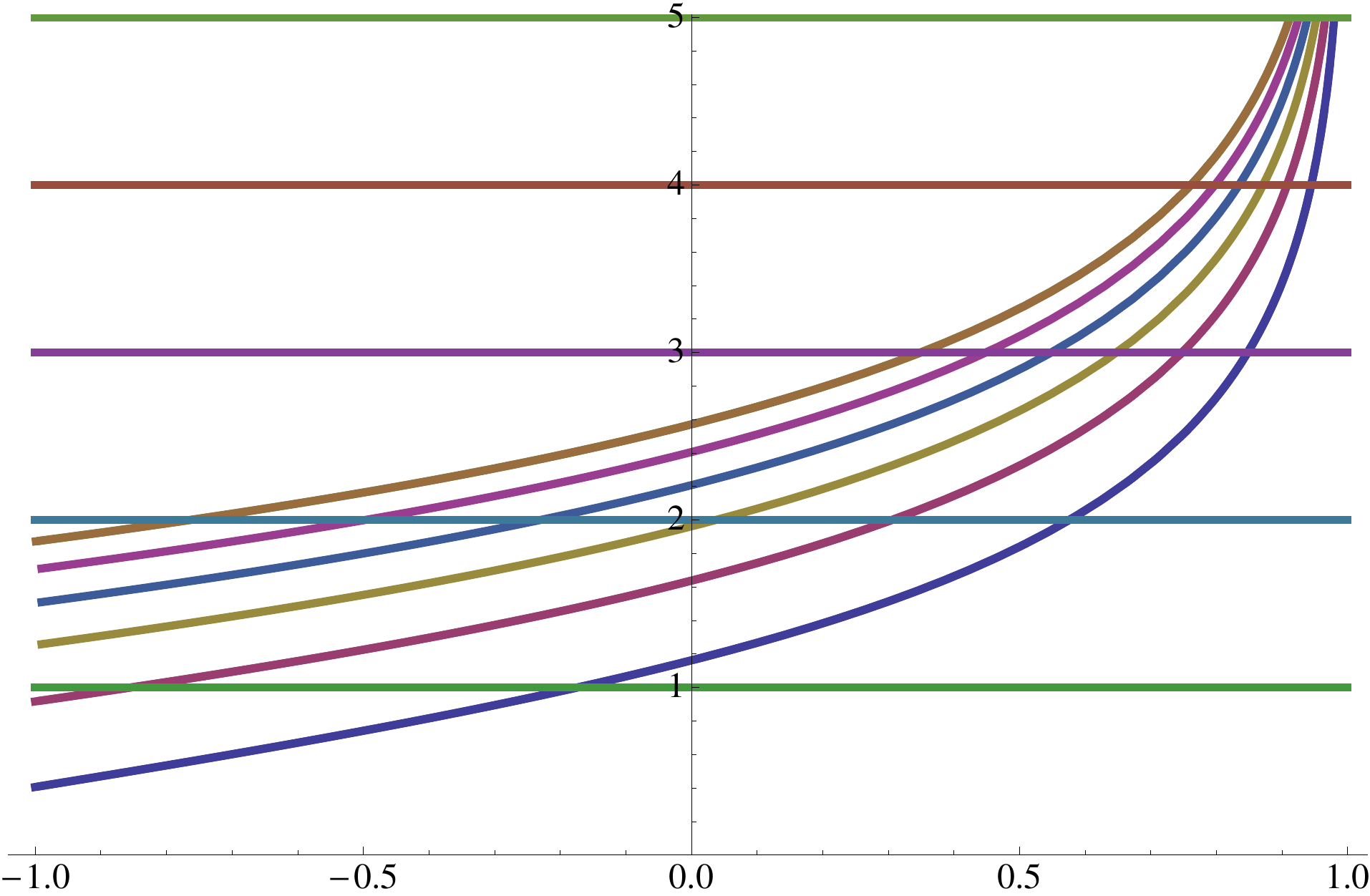}}
 \caption{{Number edge states  for different boundary conditions $\alpha=2,3,4,5$ on a 3D ball $B_3$ for  $mR_0=1$.}}
  \end{figure}

  The number of such states is always finite  and depends on the mass gap. For larger masses, there is a larger number of
  edge states.  For mass $m=\frac{1}{R_0}$, there are no edge states for
  boundary conditions with $\alpha<0.4$\,. The number of energy levels leading to edge states   for 
  different boundary conditions  $\alpha=1,2,3,4,5$  is  $3,11,18, 35, 98 $, respectively. 
  
  A final remark concerns the breaking of chiral symmetry induced by edge states. Since chiral boundary conditions 
break this symmetry, there is an induced asymmetry in the spectrum between positive and negative energy
modes. In this sense the particle-antiparticle symmetry is broken by the effect of chiral boundary conditions.

The above spinorial edge states do have a special meaning in  QCD in a three-dimensional ball with chiral boundary conditions. 
The MIT bag model uses the chiral bag boundary conditions in the limit $\alpha\to \infty$. In that case, there is an infinity of edge states.
In particular the lowest energy state is an edge state. The states of pions and protons  made of  quarks 
localised at the edges of the bag is a very interesting picture of the nucleon   where according to asymptotic freedom,  
quarks will move quite freely inside a hadron.

\section*{Acknowledgements}
M.~Asorey thanks A. Santagata for discussions. His work has been partially supported by the Spanish MICINN grants  FPA2012-35453 and CPAN Consolider Project CDS2007-42 and DGA-FSE (grant 2011-E24/2). 
A.P. Balachandran thanks the group at the Centre for High Energy Physics, Indian Institute of Science, Bangalore, and especially Sachin Vaidya for hospitality. He also thanks Andr\'es Reyes at the Universidad de los Andes for hosting him as  the Sanford Professor and for warm hospitality when this work was being completed.
J.M.~P\'erez-Pardo has been partially supported by the project MTM2014-54692, Ministry of Economy and Competitivity, Spain.

\end{document}